\documentclass[a4paper,UKenglish,cleveref, autoref, thm-restate,svgnames]{lipics-v2021}

\pdfoutput=1 
\hideLIPIcs  

\usepackage{parskip}
\usepackage{charter}

\title{Chess is Hard even for a Single Player} 

\titlerunning{Chess is Hard even for a Single Player} 

\author{N.R. Aravind}{Department of Computer Science and Engineering, Indian Institute of Technology, Hydrebad, India \and \url{https://www.iith.ac.in/~aravind/} }{aravind@cse.iith.ac.in}{}{}

\author{Neeldhara Misra}{Department of Computer Science and Engineering, Indian Institute of Technology, Gandhinagar, India \and \url{https://www.neeldhara.com} }{neeldhara.m@iitgn.ac.in}{}{}

\author{Harshil Mittal}{Department of Computer Science and Engineering, Indian Institute of Technology, Gandhinagar, India}{mittal\_harshil@iitgn.ac.in}{}{}

\authorrunning{Aravind, Misra, and Mittal} 

\Copyright{N.R. Aravind, Neeldhara Misra, and Harshil Mittal} 



\category{} 

\relatedversiondetails[]{Conference Version (to appear)}{FUN 2022} 

\funding{The second author acknowledges support from IIT Gandhinagar and the SERB-MATRICS grant MTR/2017/001033. The second and third authors also acknowledge support from the SERB-ECR grant ECR/2018/002967.}

\acknowledgements{The authors would like to thank the anonymous reviewers of FUN 2022 for their useful feedback.}

\nolinenumbers 

\EventEditors{Pierre Fraigniaud and Yushi Uno}
\EventNoEds{2}
\EventLongTitle{11th International Conference on Fun with Algorithms (FUN 2022)}
\EventShortTitle{FUN 2022}
\EventAcronym{FUN}
\EventYear{2022}
\EventDate{May 30--June 3, 2022}
\EventLocation{Island of Favignana, Sicily, Italy}
\EventLogo{}
\SeriesVolume{226}
\ArticleNo{4}

\usepackage{float}
\usepackage[backgroundcolor=SkyBlue,linecolor=SkyBlue,obeyFinal]{todonotes}

\makeatletter
\providecommand\@dotsep{5}
\def\listtodoname{}
\def\listoftodos{\@starttoc{tdo}\listtodoname}
\makeatother

\newcounter{nmcomment}

\usepackage{comment}
\excludecomment{shortver}
\includecomment{longver}

\usepackage{tikz}
\usetikzlibrary{positioning}
\usetikzlibrary{arrows}
\usetikzlibrary{snakes}
\usetikzlibrary{decorations.pathmorphing}
\usetikzlibrary{decorations.markings}
\usetikzlibrary{decorations.pathreplacing}
\usetikzlibrary{shapes.geometric}

\usepackage{latexsym,amssymb,amsmath,mathtools,eulervm,xspace}

\usepackage{skak}
\usetikzlibrary{matrix}
\usepackage{tcolorbox}
\usepackage{xcolor}

\newcommand{\NPC}{\ensuremath{\mathsf{NP}}-complete\xspace}

\usepackage{caption}
\usepackage{subcaption}

\begin{document}

\maketitle

\begin{abstract}
We introduce a generalization of ``Solo Chess'', a single-player variant of the game that can be played on chess.com. The standard version of the game is played on a regular 8 $\times$ 8 chessboard by a single player, with only white pieces, using the following rules: every move must capture a piece, no piece may capture more than 2 times, and if there is a King on the board, it must be the final piece. The goal is to clear the board, i.e, make a sequence of captures after which only one piece is left. 

We generalize this game to unbounded boards with $n$ pieces, each of which have a given number of captures that they are permitted to make. We show that \textsc{Generalized Solo Chess} is \textsf{NP}-complete, even when it is played by only rooks that have at most two captures remaining. It also turns out to be \textsf{NP}-complete even when every piece is a queen with exactly two captures remaining in the initial configuration. In contrast, we show that solvable instances of \textsc{Generalized Solo Chess} can be completely characterized when the game is: a) played by rooks on a one-dimensional board, and b) played by pawns with two captures left on a 2D board.

Inspired by \textsc{Generalized Solo Chess}, we also introduce the \textsc{Graph Capture Game}, which involves clearing a graph of tokens via captures along edges. This game subsumes \textsc{Generalized Solo Chess} played by knights. We show that the \textsc{Graph Capture Game} is \textsf{NP}-complete for undirected graphs and DAGs.
\end{abstract}

\newpage 

\section{Introduction}
\label{sec:intro}

Chess, the perfect-information two-player board game, needs to introduction. With origins dating back to as early as the 7th century, organized chess arose in the 19th century to become one of the world's most popular games in current times. At the time of this writing, the recent pandemic years witnessed a phenomenal growth of the already popular game, among spectators and amateur players alike. One of the most active computer chess sites, chess.com, is reported to have more than 75 million members, and about four million people sign in everyday.

As the reader likely knows already, chess is played on a square chessboard with 64 squares arranged in an eight-by-eight grid. At the start, each player (one controlling the white pieces, the other controlling the black pieces) controls sixteen pieces: one king, one queen, two rooks, two bishops, two knights, and eight pawns. The object of the game is to checkmate the opponent's king, whereby the king is under immediate attack (in ``check'') and there is no way for it to escape. There are also several ways a game can end in a draw. The movements of the individual pieces are subject to different constraints. While several chess engines exist for this classical version of the game, it is also known that the generalized version of chess, played on a $n \times n$ board by two players with $2n$ pieces is complete for the class \textsf{EXPTIME}~\cite{FRAENKEL1981199}. Cooperative versions of chess are also known to be hard~\cite{brunner2020complexity}.

\begin{figure}
\centering
\includegraphics[scale=0.15]{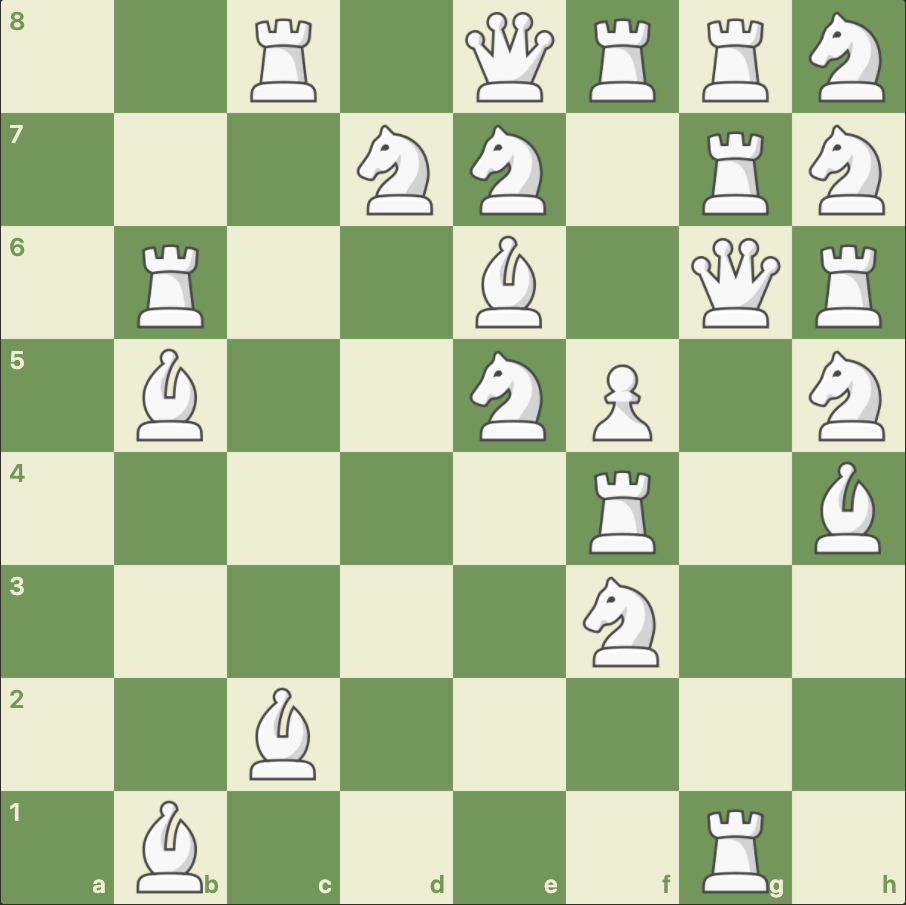}
\caption{An example of a Solo Chess configuration.}
\end{figure}

The game of \textsc{Solo Chess} is an arguably natural single-player variant of the game. We consider here a version that can be found among the chess puzzles on chess.com. The game is played on a regular 8 $\times$ 8 chessboard by a single player, with only white pieces, using the following rules: every move must capture a piece, no piece may capture more than 2 times, and if there is a King on the board, it must be the final piece. Given a board with, say, $n$ pieces in some configuration, the goal is to play a sequence of captures that ``clear'' the board. To the best of our knowledge, chess.com presents its players only with solvable configurations, even if this may not always be obvious\footnote{c.f. ``Crazy Mode''.}. The solutions, however, need not be unique. 

While the focus of our contribution here is on computational aspects of determining if a solo chess instance is solvable, we refer the reader to~\cite{ww} for a comprehensive and entertaining introduction to combinatorial game theory at large.

\paragraph*{Our Contributions}

We introduce a natural generalization of \textsc{Solo Chess} that we call \textsc{Generalized Solo Chess}$(P,d)$, where $P \subseteq \{\symqueen,\symrook,\symbishop,\sympawn,\symknight\}$ is a collection of piece types and $d \in \mathbb{N}$. This version of the game is played on the infinite integer lattice where we are given, initially, the positions of $n$ pieces, each of which is one of the types given in $P$. We are also given, for each piece, the number of captures it can make --- and further, this bound is at most $d$. The goal is to figure out if there is a sequence of $(n-1)$ valid captures such that: a) no piece captures more than the number of times it is allowed to capture; and b) the sequence of captures, when played out, ``clears the board'', i.e, only one piece remains at the end. 

We focus on settings where $|P| = 1$, i.e, when all pieces are of the same type. When the game is played only with rooks, we show that the problem is \NPC{} even when $d = 2$, but is tractable when the game is restricted to a one-dimensional board for arbitrary $d$.

\begin{restatable}[A characterization for rooks on 1D boards]{theorem}{rookseasy}
\label{thm:rooks1d}%
\textsc{Generalized Solo Chess ($\symrook,d$)} with $n$ rooks can be decided in $O(n)$ time for any $d \in \mathbb{N}$. 
\end{restatable}

\begin{restatable}[Intractability for rooks]{theorem}{rooks}
\label{thm:rooks}%
\textsc{Generalized Solo Chess ($\symrook,2$)} is \NPC{}.
\end{restatable}

When all pieces are queens, note that \textsc{Generalized Solo Chess} played on a one-dimensional board is equivalent to the game played by rooks. On the other hand, on a two-dimensional board, the game turns out to be hard even when all pieces can capture twice in the initial configuration, which is in the spirit of the regular game and is a strengthening of the hardness that we have for rooks. 

\begin{restatable}[Intractability for queens]{theorem}{queens}
\label{thm:queens}%
\textsc{Generalized Solo Chess ($\symqueen,2$)} is \NPC{} even when all queens are allowed to capture at most twice.
\end{restatable}

When all pieces are bishops, no piece has a valid move if the game is restricted to a one-dimensional board. On the other hand, it is easy to check that \textsc{Generalized Solo Chess} played on a two-dimensional board with bishops only can be reduced to \textsc{Generalized Solo Chess} played on a two-dimensional board with rooks only, by simply ``rotating'' the board 45 degrees. Therefore, we do not discuss the case of bishops explicitly. 

We now turn to the case when the game is played only with pawns: as with bishops, the game is not interesting on a one-dimensional board. However, when played on a two-dimensional board with pawns that have two captures left, it turns out that we can efficiently characterize the solvable instances. 

\begin{restatable}[A characterization for pawns]{theorem}{pawns}
\label{thm-pawn-win}
\textsc{Generalized Solo Chess ($\sympawn,2$)} with $n$ white pawns, each of which can capture at most twice, can be decided in $O(n)$ time.
\end{restatable}

When the game is played by knights only, again the game is trivial on a one-dimensional board. On a two-dimensional board, consider the following graph that is naturally associated with any configuration of knights: we introduce a vertex for every occupied position, and a pair of vertices are adjacent if and only if the corresponding positions are mutually attacking. Note that for all other pieces considered so far, an attacking pair of positions need not imply that a capture is feasible, since there may be blocking pieces in some intermediate locations. Knights are unique in that the obstacles are immaterial. This motivates the \textsc{Graph Capture} game: here we are given a graph with tokens on vertices, and the goal is to clear the tokens by a sequence of captures. The tokens can capture along edges and the number of captures that the tokens can make is given as a part of the input. 

Note that \textsc{Generalized Solo Chess}($\symknight,d$) is a special case of of \textsc{Graph Capture}$(d)$. We show that solvable instances of the latter on undirected graphs are characterized by the presence of a rooted spanning tree with the property that every internal node has at least one leaf neighbor. However, we also show that finding such spanning trees is intractable. We also show that \textsc{Graph Capture}$(d)$ is \NPC{} on DAGs. 

\begin{restatable}[Intractability of the graph capture game]{theorem}{graphcapture}
\label{thm:graphcapture} 
\textsc{Graph Capture}(2) is NP-complete on undirected graphs and DAGs even when every token can capture at most twice.
\end{restatable}

We remark that~\Cref{thm:graphcapture} has no immediate implications for {Generalized Solo Chess ($\symknight,d$)}. 

The rest of the paper is organized as follows. We establish the notation that we will use in~\Cref{sec:prelims}. The proof of~\Cref{thm:rooks,thm:rooks1d} is given in~\Cref{sec:rooks1d} and~\Cref{sec:rooks2d}, respectively. The proof of~\Cref{thm:queens} is discussed in~\Cref{sec:queens} and the proof of~\Cref{thm-pawn-win} is given in~\Cref{sec:pawns}. Finally, the proof of~\Cref{thm:graphcapture} is shown separately for undirected graphs and DAGs in~\Cref{sec:gc-undirected,sec:gc-dags}. 

\section{Preliminaries}
\label{sec:prelims}

We use $[n]$ to denote the set $\{1,2,\ldots,n\}$. We consider the following generalization\footnote{Since our focus us on the case when the game is played by pieces of one type only, we do not involve the $\symking$ in our set of pieces. Note that because of the convention that kings are never captured, any such involving only kings is trivial.} of Solo Chess. We fix a subset $P$ of $\{\symqueen,\symrook,\symbishop,\sympawn,\symknight\}$ and a positive integer $d$. The generalized game is played on an infinite two-dimensional board with $n$ pieces. For each piece, we are given an initial location and the maximum number of captures the piece is permitted to make. Such an instance is solvable if there exists a sequence $\sigma$ of $(n-1)$ valid captures with each piece making at most as many captures as it is allowed to make. We note that a capture is valid if it respects the usual rules of movements in chess. The formal definition of the problem is the following.

\begin{tcolorbox}
\textsc{Generalized Solo Chess}$(P,d)$: 

\tcblower 

{\bf Input:} A configuration $C$, which is specified by a list of $n$ triplets $(p,z,m)$, where $p \in P$, $z\in \mathbb{N}\times \mathbb{N}$, and $0 \leq m \leq d$. We use $C_i$ to refer to the $i^{th}$ triplet in $C$.

{\bf Output:} Decide if there exists a sequence of $(n-1)$ captures starting from the board position described by $C$, such that the piece corresponding to $C[i][0]$ moves at most $C[i][2]$ times for all $1 \leq i \leq n$.
\end{tcolorbox}

We note that \textsc{Generalized Solo Chess}$(P,d)$ is interesting when $d \geq 2$. Indeed, when $d = 1$, it can be efficiently determined if an instance of \textsc{Generalized Solo Chess}$(P,1)$ is solvable:

\begin{observation}
When $d=1$, a configuration $C$ is winning if and only if there's a square $z$ containing a piece, such that $z$ is reachable in one move from every other piece.
\end{observation}

\begin{proof}
The sufficiency of this condition is clear; to see the necessity, for each square $y$ on which a capture was made, let $p(y)$ be the last piece to capture on $y$. Then $p(y)$ must be the last piece standing (as it can neither move again nor be captured), and further, $y$ is the occupied square at the end of the game. Since there is exactly one occupied square at the end, this shows that all captures were made to the same square.
\end{proof}


Most of our results rely only on elementary graph-theoretic terminology and the notions of polynomial time reductions and \textsf{NP}-completeness. We refer the reader to the texts~\cite{garey1979computers,west2001introduction} for the relevant background. The well-known~\cite{garey1979computers,cygan2015parameterized}~\textsf{NP}-complete problems that we use in our reductions are the following:

\begin{enumerate}
    \item \textsc{Red-Blue Dominating Set.} Given a bipartite graph $G = (R \uplus B,E)$ and a positive integer $k$, determine if there is a subset $S \subseteq R$, $|S| \leq k$ such that $N[v] \cap S \neq \emptyset$ for all $v \in B$.  
    \item \textsc{Colorful Red-Blue Dominating Set.} Given a bipartite graph $G = (R \uplus B,E)$ where the red vertices are partitioned into $k$ disjoint parts, determine if there is a choice of exactly one vertex from each part such that every blue vertex has at least one neighbor among the chosen vertices.
    \item \textsc{3-SAT.} Given a CNF formula with at most three literals per clause, determine if there is a truth assignment to the variables that satisfies the formula. 
\end{enumerate}

\section{Solo Chess with a single piece}

\subsection{Rooks}
\label{sec:rooks}
\subsubsection{1-Dimensional boards}
\label{sec:rooks1d}
In this section we consider \textsc{Generalized Solo Chess} restricted to one-dimensional board played by rooks. It will be convenient to reason about such instances by using strings to represent game configurations; to this end we introduce some terminology. 




\begin{definition}[Configuration]
\label{defn:configuration}
A \emph{configuration} is string $s$ over $\{0,1,2,\ldots,d,\square\}$. It denotes a board of size $1 \times N$, where $N$ is the length of the string. The cell $(1,j)$ is empty if $s[j] = \square$, and is otherwise occupied by a rook with $b$ moves left where $b := s[j]$. 
\end{definition}

We refer the reader to~\Cref{fig:example1} for an example and how a given board position translates to a configuration as defined above. Informally, a configuration is \emph{solvable} if there is a valid sequence of moves that clears the board.

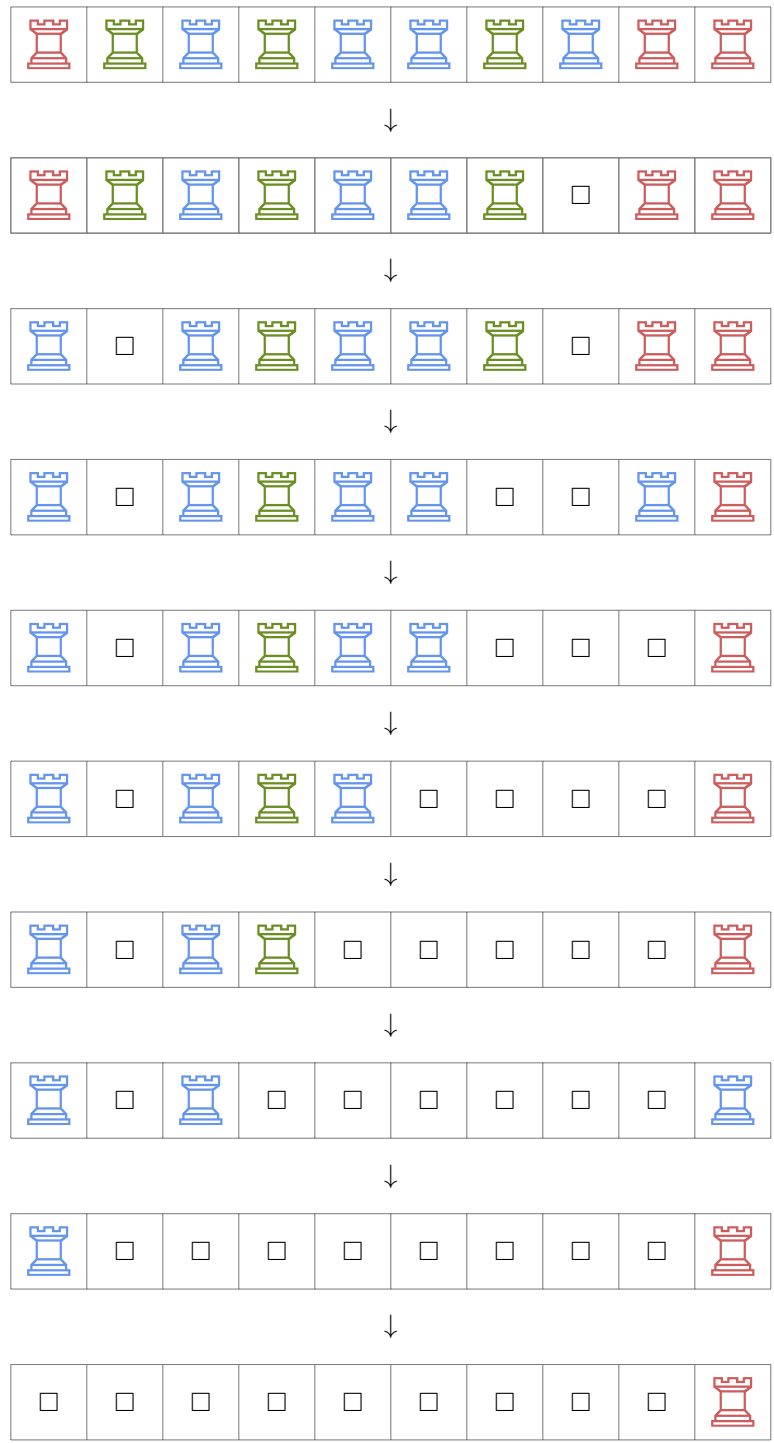
\begin{figure}
    \centering
    \begin{tikzpicture}


\draw[step=1cm,color=gray] (0,20) grid (10,19);
\node at (0.5,19.5) {{\color{IndianRed}{\Huge$\symrook$}}}; 
\node at (1.5,19.5) {{\color{OliveDrab}{\Huge$\symrook$}}}; 
\node at (2.5,19.5) {{\color{CornflowerBlue}{\Huge$\symrook$}}}; 
\node at (3.5,19.5) {{\color{OliveDrab}{\Huge$\symrook$}}}; 
\node at (4.5,19.5) {{\color{CornflowerBlue}{\Huge$\symrook$}}}; 
\node at (5.5,19.5) {{\color{CornflowerBlue}{\Huge$\symrook$}}}; 
\node at (6.5,19.5) {{\color{OliveDrab}{\Huge$\symrook$}}}; 
\node at (7.5,19.5) {{\color{CornflowerBlue}{\Huge$\symrook$}}}; 
\node at (8.5,19.5) {{\color{IndianRed}{\Huge$\symrook$}}}; 
\node at (9.5,19.5) {{\color{IndianRed}{\Huge$\symrook$}}}; 

\node at (5,18.5) {$\downarrow$}; 

\draw[step=1cm,color=gray] (0,18) grid (10,17);
\node at (0.5,17.5) {{\color{IndianRed}{\Huge$\symrook$}}}; 
\node at (1.5,17.5) {{\color{OliveDrab}{\Huge$\symrook$}}}; 
\node at (2.5,17.5) {{\color{CornflowerBlue}{\Huge$\symrook$}}}; 
\node at (3.5,17.5) {{\color{OliveDrab}{\Huge$\symrook$}}}; 
\node at (4.5,17.5) {{\color{CornflowerBlue}{\Huge$\symrook$}}}; 
\node at (5.5,17.5) {{\color{CornflowerBlue}{\Huge$\symrook$}}}; 
\node at (6.5,17.5) {{\color{OliveDrab}{\Huge$\symrook$}}}; 
\node at (7.5,17.5) {$\square$}; 
\node at (8.5,17.5) {{\color{IndianRed}{\Huge$\symrook$}}}; 
\node at (9.5,17.5) {{\color{IndianRed}{\Huge$\symrook$}}};

\node at (5,16.5) {$\downarrow$}; 

\draw[step=1cm,color=gray] (0,16) grid (10,15);

\draw[step=1cm,color=gray] (0,18) grid (10,17);
\node at (0.5,15.5) {{\color{CornflowerBlue}{\Huge$\symrook$}}}; 
\node at (1.5,15.5) {$\square$}; 
\node at (2.5,15.5) {{\color{CornflowerBlue}{\Huge$\symrook$}}}; 
\node at (3.5,15.5) {{\color{OliveDrab}{\Huge$\symrook$}}}; 
\node at (4.5,15.5) {{\color{CornflowerBlue}{\Huge$\symrook$}}}; 
\node at (5.5,15.5) {{\color{CornflowerBlue}{\Huge$\symrook$}}}; 
\node at (6.5,15.5) {{\color{OliveDrab}{\Huge$\symrook$}}}; 
\node at (7.5,15.5) {$\square$}; 
\node at (8.5,15.5) {{\color{IndianRed}{\Huge$\symrook$}}}; 
\node at (9.5,15.5) {{\color{IndianRed}{\Huge$\symrook$}}};

\node at (5,14.5) {$\downarrow$}; 

\draw[step=1cm,color=gray] (0,14) grid (10,13);

\node at (0.5,13.5) {{\color{CornflowerBlue}{\Huge$\symrook$}}}; 
\node at (1.5,13.5) {$\square$}; 
\node at (2.5,13.5) {{\color{CornflowerBlue}{\Huge$\symrook$}}}; 
\node at (3.5,13.5) {{\color{OliveDrab}{\Huge$\symrook$}}}; 
\node at (4.5,13.5) {{\color{CornflowerBlue}{\Huge$\symrook$}}}; 
\node at (5.5,13.5) {{\color{CornflowerBlue}{\Huge$\symrook$}}}; 
\node at (6.5,13.5) {$\square$}; 
\node at (7.5,13.5) {$\square$}; 
\node at (8.5,13.5) {{\color{CornflowerBlue}{\Huge$\symrook$}}}; 
\node at (9.5,13.5) {{\color{IndianRed}{\Huge$\symrook$}}};

\node at (5,12.5) {$\downarrow$}; 

\draw[step=1cm,color=gray] (0,12) grid (10,11);

\node at (0.5,11.5) {{\color{CornflowerBlue}{\Huge$\symrook$}}}; 
\node at (1.5,11.5) {$\square$}; 
\node at (2.5,11.5) {{\color{CornflowerBlue}{\Huge$\symrook$}}}; 
\node at (3.5,11.5) {{\color{OliveDrab}{\Huge$\symrook$}}}; 
\node at (4.5,11.5) {{\color{CornflowerBlue}{\Huge$\symrook$}}}; 
\node at (5.5,11.5) {{\color{CornflowerBlue}{\Huge$\symrook$}}}; 
\node at (6.5,11.5) {$\square$}; 
\node at (7.5,11.5) {$\square$}; 
\node at (8.5,11.5) {$\square$}; 
\node at (9.5,11.5) {{\color{IndianRed}{\Huge$\symrook$}}};

\node at (5,10.5) {$\downarrow$}; 

\draw[step=1cm,color=gray] (0,10) grid (10,9);

\node at (0.5,9.5) {{\color{CornflowerBlue}{\Huge$\symrook$}}}; 
\node at (1.5,9.5) {$\square$}; 
\node at (2.5,9.5) {{\color{CornflowerBlue}{\Huge$\symrook$}}}; 
\node at (3.5,9.5) {{\color{OliveDrab}{\Huge$\symrook$}}}; 
\node at (4.5,9.5) {{\color{CornflowerBlue}{\Huge$\symrook$}}}; 
\node at (5.5,9.5) {$\square$}; 
\node at (6.5,9.5) {$\square$}; 
\node at (7.5,9.5) {$\square$}; 
\node at (8.5,9.5) {$\square$}; 
\node at (9.5,9.5) {{\color{IndianRed}{\Huge$\symrook$}}};

\node at (5,8.5) {$\downarrow$}; 

\draw[step=1cm,color=gray] (0,8) grid (10,7);

\node at (0.5,7.5) {{\color{CornflowerBlue}{\Huge$\symrook$}}}; 
\node at (1.5,7.5) {$\square$}; 
\node at (2.5,7.5) {{\color{CornflowerBlue}{\Huge$\symrook$}}}; 
\node at (3.5,7.5) {{\color{OliveDrab}{\Huge$\symrook$}}}; 
\node at (4.5,7.5) {$\square$}; 
\node at (5.5,7.5) {$\square$}; 
\node at (6.5,7.5) {$\square$}; 
\node at (7.5,7.5) {$\square$}; 
\node at (8.5,7.5) {$\square$}; 
\node at (9.5,7.5) {{\color{IndianRed}{\Huge$\symrook$}}};

\node at (5,6.5) {$\downarrow$}; 

\draw[step=1cm,color=gray] (0,6) grid (10,5);

\node at (0.5,5.5) {{\color{CornflowerBlue}{\Huge$\symrook$}}}; 
\node at (1.5,5.5) {$\square$}; 
\node at (2.5,5.5) {{\color{CornflowerBlue}{\Huge$\symrook$}}}; 
\node at (3.5,5.5) {$\square$}; 
\node at (4.5,5.5) {$\square$}; 
\node at (5.5,5.5) {$\square$}; 
\node at (6.5,5.5) {$\square$}; 
\node at (7.5,5.5) {$\square$}; 
\node at (8.5,5.5) {$\square$}; 
\node at (9.5,5.5) {{\color{CornflowerBlue}{\Huge$\symrook$}}};

\node at (5,4.5) {$\downarrow$}; 

\draw[step=1cm,color=gray] (0,4) grid (10,3);

\node at (0.5,3.5) {{\color{CornflowerBlue}{\Huge$\symrook$}}}; 
\node at (1.5,3.5) {$\square$}; 
\node at (2.5,3.5) {$\square$}; 
\node at (3.5,3.5) {$\square$}; 
\node at (4.5,3.5) {$\square$}; 
\node at (5.5,3.5) {$\square$}; 
\node at (6.5,3.5) {$\square$}; 
\node at (7.5,3.5) {$\square$}; 
\node at (8.5,3.5) {$\square$}; 
\node at (9.5,3.5) {{\color{IndianRed}{\Huge$\symrook$}}};

\node at (5,2.5) {$\downarrow$}; 

\draw[step=1cm,color=gray] (0,2) grid (10,1);

\node at (0.5,1.5) {$\square$}; 
\node at (1.5,1.5) {$\square$}; 
\node at (2.5,1.5) {$\square$}; 
\node at (3.5,1.5) {$\square$}; 
\node at (4.5,1.5) {$\square$}; 
\node at (5.5,1.5) {$\square$}; 
\node at (6.5,1.5) {$\square$}; 
\node at (7.5,1.5) {$\square$}; 
\node at (8.5,1.5) {$\square$}; 
\node at (9.5,1.5) {{\color{IndianRed}{\Huge$\symrook$}}};

\end{tikzpicture}
    \caption{An example of a valid sequence of captures that clears the board. The initial configuration corresponds to the string \texttt{0212112100}. In other words, the red, blue, and green rooks denote rooks with zero, one, and two moves left, respectively. Notice that this is not a unique solution --- there are several other valid sequences that also successfully clear this board.}
    \label{fig:example1}
\end{figure}
\begin{definition}[$\ell$-solvable configuration]
\label{defn:ell-solvableconfiguration} Let $s$ be a configuration of length $N$ and let $1\leq \ell\leq N$. We say that $s$ is $\ell$-solvable if there exists a sequence of moves that clears the corresponding board with the final rook at the cell $(1,\ell)$.
\end{definition}

Our main goal in this section is to establish the following:

\rookseasy*

Note that for any sequence (say $\sigma$) of moves that clears the board with final rook at the cell $(1,\ell)$, no move of $\sigma$ empties the cell $(1,\ell)$ and thus, there's no move of $\sigma$ wherein the cells containing the captured piece and the capturing piece are at different sides, i.e., one at left and the other at right, of the cell $(1,\ell)$. So, note that $s$ is $\ell$-solvable if and only if there is a position such that the sub-configurations to the left and right of location $\ell$ are independently solvable. We now develop a criteria for solving 1D configurations where the target piece is one of the extreme locations on the board. Note that when $d = 2$, the first criteria below amounts to saying that $s$ is $N$-solvable iff $s[N]\neq\square$ and $s[1,\ldots,N-1]$ has at least as many $2$'s as $0$'s; and the second criteria states that $s$ is $1$-solvable iff $s[1]\neq\square$ and $s[2,\ldots,N]$ has at least as many $2$'s as $0$'s. A direct proof of this simpler statement is given in the Appendix.

\begin{lemma}
\label{lem:rooks1Dgen}
For every configuration $s$ of length $N$,
\begin{enumerate}
    \item $s$ is $N$-solvable iff $s[N]\neq\square$ and $\underset{\substack{1\leq i\leq N-1:\\s[i]\not\in\{0,\square\}}}{\sum}\big(s[i]-1\big) \geq$ number of $0$'s in $s[1,\ldots,N-1]$
    \item $s$ is $1$-solvable iff $s[1]\neq\square$ and $\underset{\substack{2\leq i\leq N:\\s[i]\not\in\{0,\square\}}}{\sum}\big(s[i]-1\big) \geq$ number of $0$'s in $s[2,\ldots,N]$
\end{enumerate}
\end{lemma}
\begin{proof} We argue the first claim since the proof of the second is symmetric.
For the forward implication, we show (using induction on $m$) that the following statement is true for all integers $m\geq 0$: For every configuration $s$ such that $\underset{\substack{1\leq i\leq N-1:\\s[i]\not\in\{0,\square\}}}{\sum}\big(s[i]-1\big) = m$, if $s$ is $N$-solvable, then $s[N]\neq\square$ and $m\geq$ number of $0$'s in $s[1,\ldots,N-1]$. For the base case, consider $m=0$. The only configurations $s$ with $\underset{\substack{1\leq i\leq N-1:\\s[i]\not\in\{0,\square\}}}{\sum}\big(s[i]-1\big) = 0$ are the ones for which $s[1,\ldots,N-1]$ is a string over $\{0,1,\square\}$.  Among these, the only $N$-solvable configurations $s$ are the ones for which $s[1,\ldots,N-1]$ is a string over $\{1,\square\}$ and $s[N]\neq\square$. Thus, the statement is true for $m=0$.\\\\
As induction hypothesis, assume that the statement is true for all integers $0\leq m\leq p$, for some integer $p\geq 0$. Let's argue that the statement is true for $m=p+1$. Let $s$ be a configuration such that $\underset{\substack{1\leq i\leq N-1:\\s[i]\not\in\{0,\square\}}}{\sum}\big(s[i]-1\big)=p+1$ and $s$ is $N$-solvable. As $s$ is $N$-solvable, there exists a sequence of moves (say $\sigma$) that clears the corresponding $1\times N$ board with the final rook at the cell $(1,N)$. Clearly, $s[N]\neq\square$. Let $t\geq 1$ be the least integer such that the capturing piece in $t^{th}$ move of $\sigma$ is not a $1$-rook. Let $\tilde{s}$ denote the configuration corresponding to the board obtained after $t^{th}$ move of $\sigma$. Note that $\tilde{s}$ is $N$-solvable and $\underset{\substack{1\leq i\leq N-1:\\\tilde{s}[i]\not\in\{0,\square\}}}{\sum}(\tilde{s}[i]-1)\leq p$. Using induction hypothesis, $p\geq$ number of $0's$ in $\tilde{s}[1,\ldots,N-1]$. Also, number of $0$'s in $\tilde{s}[1,\ldots,N-1]\geq$ number of $0$'s in $s[1,\ldots,N-1]-1$; this is because the number of $0$-rooks at the cells $(1,1),\ldots,(1,N-1)$ doesn't decrease in the first $t-1$ moves of $\sigma$, and decreases by at most $1$ in the $t^{th}$ move of $\sigma$. Therefore, we have $p+1\geq$ number of $0$'s in $s[1,\ldots,N-1]$, as desired. 
\\\\
For the converse, we show (using induction on $m$) that the following statement is true for all integers $m\geq 0$: For every configuration $s$ such that $s[N]\neq\square$, if $s[1,\ldots,N-1]$ has exactly $m$ $0's$ and $\underset{\substack{1\leq i\leq N-1:\\s[i]\not\in\{0,\square\}}}{\sum}\big(s[i]-1\big) \geq m$ , then $s$ is $N$-solvable. 

For the base case, consider $m=0$. Let $s$ be a configuration such that $s[N]\neq\square$ and $s[1,\ldots,N-1]$ has no $0$'s. For each $1\leq i< N$, the cell $(1,i)$ in the corresponding board is either empty or has a $1/2/\ldots/d$-rook. The board can be cleared with the final piece at $(1,N)$ by making the $1/2/\ldots/d$-rooks (if any) to capture rook at the cell $(1,N)$. Thus, $s$ is $N$-solvable. So, the statement is true for $m=0$.\\\\
As induction hypothesis, assume that the statement is true for all integers $0\leq m\leq p$, for some integer $p\geq 0$. Let's argue that the statement is true for $m=p+1$. Let $s$ be a configuration such that $s[N]\neq\square$, $s[1,\ldots,N-1]$ has exactly $(p+1)$ $0's$ and $\underset{\substack{1\leq i\leq N-1:\\s[i]\not\in\{0,\square\}}}{\sum}\big(s[i]-1\big) \geq p+1$. While there is a $1$-rook in the cells $(1,1),\ldots,(1,N-1)$ that can capture a $0$-rook in the cells $(1,1),\ldots,(1,N-1)$, make such a move. Once no such move can be made, there exist integers $1\leq u<v< N$ such that $s[u,\ldots,v]= 0\square^{\lambda}x$ or $s[u,\ldots,v]=x\square^{\lambda}0$, for some $\lambda\geq 0$ and some $2\leq x\leq d$. In the former (resp. latter) case, the $x$-rook at the cell $(1,v)$ (resp. $(1,u)$) can be made to capture the $0$-rook at the cell $(1,u)$ (resp. $(1,v)$), and the configuration corresponding to the resulting board is $N$-solvable by induction hypothesis.
\end{proof}

We conclude that a configuration $s$ of length $N$ is solvable iff there exists $1\leq \ell\leq N$ such that
\begin{itemize}
    \item $s[\ell]\neq\square$,
    \item $\underset{\substack{1\leq i\leq \ell-1:\\s[i]\not\in\{0,\square\}}}{\sum}\big(s[i]-1\big) \geq$ number of $0$'s in $s[1,\ldots,\ell-1]$, and
    \item $\underset{\substack{\ell+1\leq i\leq N:\\s[i]\not\in\{0,\square\}}}{\sum}\big(s[i]-1\big) \geq$ number of $0$'s in $s[\ell+1,\ldots,N]$.
\end{itemize}

\subsubsection{2-Dimensional boards}
\label{sec:rooks2d}

\rooks*

\begin{proof}

We reduce from the \textsc{Red-Blue Dominating Set} problem. Let $\mathcal{I} := \langle G = (N \cup T, E); k \rangle$ be an instance of \textsc{Red-Blue Dominating Set}. Recall that $G$ is a bipartite graph with bipartition $N$ and $T$; and $\mathcal{I}$ is a \textsc{Yes}-instance if and only if there exists a subset $S \subseteq N$ of size at most $k$ such that every vertex $v$ in $T$ has a neighbor in $S$. We let the vertices in $N$ be denoted by $[n]$ and let $T := \{v_1, \ldots, v_m\}$. We refer to the vertices of $N$ and $T$ as non-terminals and terminals, respectively.

We first describe the construction of the reduced instance of \textsc{Generalized Solo Chess ($\symrook,2$)} based on $\mathcal{I}$. The game takes place on a $(2m+1) \times (n+m+k+1)$ board. The initial position of the rooks is as follows:

\begin{itemize}
    \item \emph{Non-terminal} rooks. For all $i \in [n]$, we place a $1$-rook in the cell $(2m+1,i)$.
    \item \emph{Terminal} rooks. For all $j \in [m]$, we place a $1$-rook in the cell $(2j-1,\ell)$ for each $\ell$ such that $\ell \in N(v_j)$. 
    \item \emph{Collector} rooks. For all $j \in [m]$, we place a $2$-rook in the cell $(2j-1,n+j)$.
    \item \emph{Cleaner} rooks. For all $\ell \in [k]$, we place a $2$-rook in the cell $(2m+1,n+m+\ell)$.
    \item Target location. Finally, we place on $1$-rook at the location $(2m+1,n+m+k+1)$.
\end{itemize}

The non-terminal and terminal rooks correspond to the non-terminal and terminal vertices in the graph, and their relative positioning as described above captures the graph structure. The rooks on every row are expected to ``clear to one of the columns corresponding to a vertex they are dominated by'', and the other auxiliary rooks added to the board above help with clearing the board after this phase, as explained further below.

\begin{figure}
    \centering
    \includegraphics[scale=0.8]{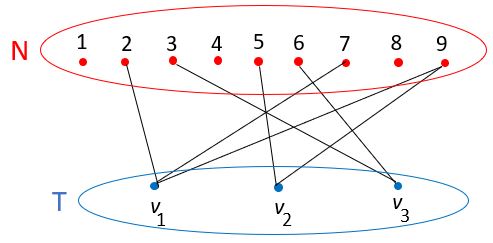}
    \caption{An instance of Red-Blue Dominating Set.}
    \label{fig:rbds}
\end{figure}

\begin{figure}
    \centering
    $$~$$
    \begin{tikzpicture}[scale=0.87]
    
        
        \draw[step=1cm,color=gray] (0,9) grid (15,2);
        
        \node at (14.5,8.5) {{\color{OliveDrab}{$\star$}}}; 
        \node at (13.5,8.5) {{\color{OliveDrab}{\Huge$\symrook$}}}; 
        \node at (12.5,8.5) {{\color{OliveDrab}{\Huge$\symrook$}}}; 
        
        \node at (0.5,8.5) {{\color{CornflowerBlue}{\Huge$\symrook$}}}; 
        \node at (1.5,8.5) {{\color{CornflowerBlue}{\Huge$\symrook$}}}; 
        \node at (2.5,8.5) {{\color{CornflowerBlue}{\Huge$\symrook$}}}; 
        
        \node at (0.5,9.5) {$1$}; 
        \node at (1.5,9.5) {$2$}; 
        \node at (2.5,9.5) {$3$}; 
        
        \node at (3.5,8.5) {{\color{CornflowerBlue}{\Huge$\symrook$}}}; 
        \node at (4.5,8.5) {{\color{CornflowerBlue}{\Huge$\symrook$}}}; 
        \node at (5.5,8.5) {{\color{CornflowerBlue}{\Huge$\symrook$}}}; 
        
        \node at (3.5,9.5) {$4$}; 
        \node at (4.5,9.5) {$5$}; 
        \node at (5.5,9.5) {$6$}; 
        
        \node at (6.5,8.5) {{\color{CornflowerBlue}{\Huge$\symrook$}}}; 
        \node at (7.5,8.5) {{\color{CornflowerBlue}{\Huge$\symrook$}}}; 
        \node at (8.5,8.5) {{\color{CornflowerBlue}{\Huge$\symrook$}}}; 
        
        \node at (6.5,9.5) {$7$}; 
        \node at (7.5,9.5) {$8$}; 
        \node at (8.5,9.5) {$9$};

        \node at (2.5,6.5)  {{\color{CornflowerBlue}{\Huge$\symrook$}}}; 
        \node at (5.5,6.5)  {{\color{CornflowerBlue}{\Huge$\symrook$}}}; 
        \node at (9.5,2.5)  {{\color{OliveDrab}{\Huge$\symrook$}}}; 
        
        \node at(-0.5,6.5) {$v_3$};

        \node at (4.5,4.5)  {{\color{CornflowerBlue}{\Huge$\symrook$}}}; 
        \node at (8.5,4.5)  {{\color{CornflowerBlue}{\Huge$\symrook$}}}; 
        \node at (10.5,4.5)  {{\color{OliveDrab}{\Huge$\symrook$}}}; 
        
        \node at(-0.5,4.5) {$v_2$}; 
        
        \node at (1.5,2.5)  {{\color{CornflowerBlue}{\Huge$\symrook$}}}; 
        \node at (6.5,2.5)  {{\color{CornflowerBlue}{\Huge$\symrook$}}}; 
        \node at (8.5,2.5)  {{\color{CornflowerBlue}{\Huge$\symrook$}}}; 
        \node at (11.5,6.5)  {{\color{OliveDrab}{\Huge$\symrook$}}}; 
        
        \node at(-0.5,2.5) {$v_1$}; 

\end{tikzpicture}
    \caption{The reduced instance of~\textsc{Generalized Solo Chess} played by rooks corresponding to the instance shown in~\Cref{fig:rbds}.}
    \label{fig:example1redn}
\end{figure}
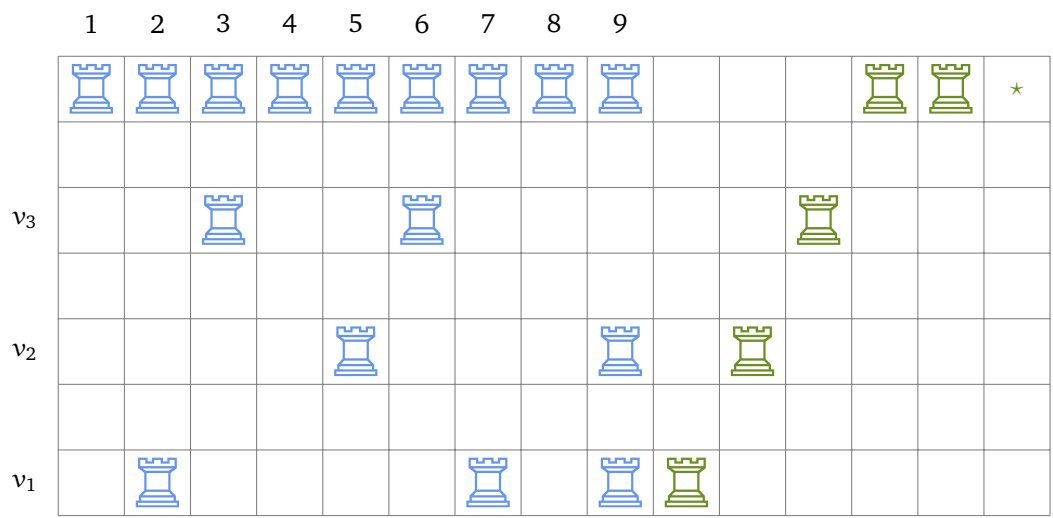

\paragraph*{Forward Direction}
Assume that $\mathcal{I}$ is a YES-instance. That is, there exists $S\subseteq [n]$ of size at most $k$ such that every vertex in $T$ has a neighbour in $S$. For each $j\in[m]$, let $f(j)$ denote an arbitrary but fixed $i\in S$ such that $i\in N(v_j)$. Consider the following sequence of moves (c.f.~\Cref{fig:threegraphs}):

\begin{itemize}
    \item For each $j\in[m]$ and each $\ell\in N(v_j)\setminus\{f(j)\}$, the terminal $1$-rook at the cell $(2j-1,\ell)$ captures rook at the cell $(2j-1,f(j))$.
    \item For each $j\in[m]$, the collector $2$-rook at the cell $(2j-1,n+j)$ captures the rook at the cell $(2j-1,f(j))$, and the $1$-rook at the cell $(2j-1,f(j))$ so obtained then captures rook at the cell $(2m+1,f(j))$.
    \item For each $i\in[n]$, if there's a non-terminal $1$-rook at the cell $(2m+1,i)$, then it captures one of the $0$-rooks at the cells $(2m+1,f(1)),\ldots,(2m+1,f(m))$.
\end{itemize}
Now, the board is empty except for the top row which has one $1$-rook at the target location, $k$ cleaner $2$-rooks and at most $k$ $0$-rooks, i.e., the $0$-rooks at the cells $(2m+1,f(1)),\ldots,(2m+1,f(m))$. Using Lemma \ref{lem:rooks1Dgen}, this corresponds to a $(n+m+k+1)$-solvable configuration.

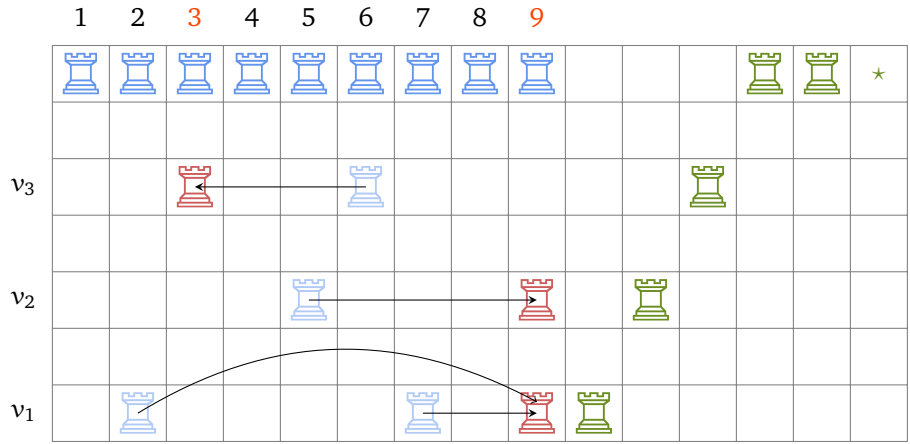
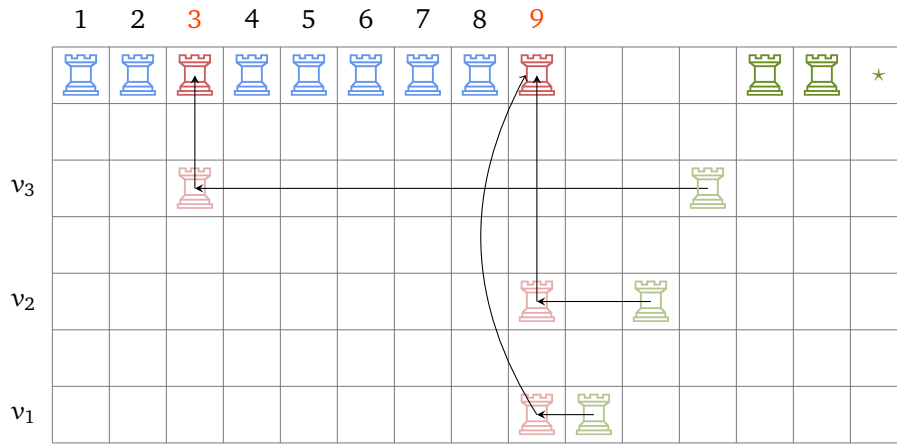
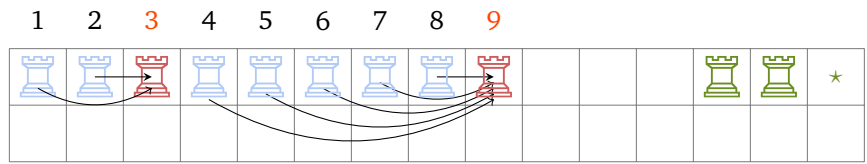
\begin{figure}
     \centering
     \begin{subfigure}[b]{\textwidth}
        \centering
         
        \begin{tikzpicture}[scale=0.75]

        \draw[step=1cm,color=gray] (0,9) grid (15,2);
        
        \node at (14.5,8.5) {{\color{OliveDrab}{$\star$}}}; 
        \node at (13.5,8.5) {{\color{OliveDrab}{\huge$\symrook$}}}; 
        \node at (12.5,8.5) {{\color{OliveDrab}{\huge$\symrook$}}}; 
        
        \node at (0.5,8.5) {{\color{CornflowerBlue}{\huge$\symrook$}}}; 
        \node at (1.5,8.5) {{\color{CornflowerBlue}{\huge$\symrook$}}}; 
        \node at (2.5,8.5) {{\color{CornflowerBlue}{\huge$\symrook$}}}; 
        
        \node at (0.5,9.5) {$1$}; 
        \node at (1.5,9.5) {$2$}; 
        \node at (2.5,9.5) {${\color{OrangeRed}3}$}; 
        
        \node at (3.5,8.5) {{\color{CornflowerBlue}{\huge$\symrook$}}}; 
        \node at (4.5,8.5) {{\color{CornflowerBlue}{\huge$\symrook$}}}; 
        \node at (5.5,8.5) {{\color{CornflowerBlue}{\huge$\symrook$}}}; 
        
        \node at (3.5,9.5) {$4$}; 
        \node at (4.5,9.5) {$5$}; 
        \node at (5.5,9.5) {$6$}; 
        
        \node at (6.5,8.5) {{\color{CornflowerBlue}{\huge$\symrook$}}}; 
        \node at (7.5,8.5) {{\color{CornflowerBlue}{\huge$\symrook$}}}; 
        \node at (8.5,8.5) {{\color{CornflowerBlue}{\huge$\symrook$}}}; 
        
        \node at (6.5,9.5) {$7$}; 
        \node at (7.5,9.5) {$8$}; 
        \node at (8.5,9.5) {${\color{OrangeRed}9}$};

        \node at (2.5,6.5)  {{\color{IndianRed}{\huge$\symrook$}}}; 
        \node at (5.5,6.5)  {{\color{CornflowerBlue!50}{\huge$\symrook$}}};
        
        \draw [stealth-](2.5,6.5) -- (5.5,6.5);
        
        \node at (9.5,2.5)  {{\color{OliveDrab}{\huge$\symrook$}}}; 
        
        \node at(-0.5,6.5) {$v_3$};

        \node at (4.5,4.5)  {{\color{CornflowerBlue!50}{\huge$\symrook$}}}; 
        \node at (8.5,4.5)  {{\color{IndianRed}{\huge$\symrook$}}};
        
        \draw [-stealth](4.5,4.5) -- (8.5,4.5);
        
        \node at (10.5,4.5)  {{\color{OliveDrab}{\huge$\symrook$}}}; 
        
        \node at(-0.5,4.5) {$v_2$}; 
        
        \node at (1.5,2.5)  {{\color{CornflowerBlue!50}{\huge$\symrook$}}}; 
        \node at (6.5,2.5)  {{\color{CornflowerBlue!50}{\huge$\symrook$}}}; 
        \node at (8.5,2.5)  {{\color{IndianRed}{\huge$\symrook$}}};
        
        \draw [-stealth](6.5,2.5) -- (8.5,2.5);
        
        \draw[->] (1.5,2.5) to[bend left] (8.5,2.7);
        
        \node at (11.5,6.5)  {{\color{OliveDrab}{\huge$\symrook$}}}; 
        
        \node at(-0.5,2.5) {$v_1$}; 

\end{tikzpicture}
         
        \caption{All blue rooks on rows corresponding to blue vertices clear row-wise to the dominating set vertices.}
        \label{fig:fwd-rooks-step1}
     \end{subfigure}
     \hfill
     \begin{subfigure}[b]{\textwidth}
        \centering
         
        \begin{tikzpicture}[scale=0.75]

 \draw[step=1cm,color=gray] (0,9) grid (15,2);
        
        \node at (14.5,8.5) {{\color{OliveDrab}{$\star$}}}; 
        \node at (13.5,8.5) {{\color{OliveDrab}{\huge$\symrook$}}}; 
        \node at (12.5,8.5) {{\color{OliveDrab}{\huge$\symrook$}}}; 
        
        \node at (0.5,8.5) {{\color{CornflowerBlue}{\huge$\symrook$}}}; 
        \node at (1.5,8.5) {{\color{CornflowerBlue}{\huge$\symrook$}}}; 
        \node at (2.5,8.5) {{\color{IndianRed}{\huge$\symrook$}}}; 
        
        \node at (0.5,9.5) {$1$}; 
        \node at (1.5,9.5) {$2$}; 
        \node at (2.5,9.5) {${\color{OrangeRed}3}$}; 
        
        \node at (3.5,8.5) {{\color{CornflowerBlue}{\huge$\symrook$}}}; 
        \node at (4.5,8.5) {{\color{CornflowerBlue}{\huge$\symrook$}}}; 
        \node at (5.5,8.5) {{\color{CornflowerBlue}{\huge$\symrook$}}}; 
        
        \node at (3.5,9.5) {$4$}; 
        \node at (4.5,9.5) {$5$}; 
        \node at (5.5,9.5) {$6$}; 
        
        \node at (6.5,8.5) {{\color{CornflowerBlue}{\huge$\symrook$}}}; 
        \node at (7.5,8.5) {{\color{CornflowerBlue}{\huge$\symrook$}}}; 
        \node at (8.5,8.5) {{\color{IndianRed}{\huge$\symrook$}}}; 
        
        \node at (6.5,9.5) {$7$}; 
        \node at (7.5,9.5) {$8$}; 
        \node at (8.5,9.5) {${\color{OrangeRed}9}$};

        \node at (2.5,6.5)  {{\color{IndianRed!50}{\huge$\symrook$}}}; 
        \node at (9.5,2.5)  {{\color{OliveDrab!50}{\huge$\symrook$}}}; 
        
        \draw [stealth-](2.5,6.5) -- (11.5,6.5);
        
        \draw [stealth-](2.5,8.5) -- (2.5,6.5);
        
        \node at(-0.5,6.5) {$v_3$};

        \node at (8.5,4.5)  {{\color{IndianRed!50}{\huge$\symrook$}}}; 
        \node at (10.5,4.5)  {{\color{OliveDrab!50}{\huge$\symrook$}}}; 
        
        \draw [stealth-](8.5,4.5) -- (10.5,4.5);
        
        \node at(-0.5,4.5) {$v_2$}; 
        
        \node at (8.5,2.5)  {{\color{IndianRed!50}{\huge$\symrook$}}}; 
        \node at (11.5,6.5)  {{\color{OliveDrab!50}{\huge$\symrook$}}}; 
        
        \draw [stealth-](8.5,2.5) -- (9.5,2.5);
        
        \draw [stealth-](8.5,8.5) -- (8.5,4.5);
        
        \draw[->] (8.5,2.5) to[bend left] (8.3,8.5);
        
        \node at(-0.5,2.5) {$v_1$};
        \end{tikzpicture}
         
        \caption{The green rooks on the rows corresponding to blue vertices ``pick up'' the red rooks and capture along the column to get the rook on the top row.}
        \label{fig:fwd-rooks-step2}
     \end{subfigure}
     \hfill
     \begin{subfigure}[b]{\textwidth}
        \centering
         
        \begin{tikzpicture}[scale=0.75]
    
        
        \draw[step=1cm,color=gray] (0,10) grid (15,8);
        
        \node at (14.5,9.5) {{\color{OliveDrab}{$\star$}}}; 
        \node at (13.5,9.5) {{\color{OliveDrab}{\huge$\symrook$}}}; 
        \node at (12.5,9.5) {{\color{OliveDrab}{\huge$\symrook$}}}; 
        
        \node at (0.5,9.5) {{\color{CornflowerBlue!50}{\huge$\symrook$}}}; 
        \node at (1.5,9.5) {{\color{CornflowerBlue!50}{\huge$\symrook$}}}; 
        \node at (2.5,9.5) {{\color{IndianRed}{\huge$\symrook$}}}; 
        
        \draw[->] (0.5,9.3) to[bend right] (2.5,9.3);
        
        \draw[->] (6.5,9.4) to[bend right] (8.5,9.4);
        
        \draw[->] (5.5,9.3) to[bend right] (8.5,9.3);
        
        \draw[->] (4.5,9.2) to[bend right] (8.5,9.2);
        
        \draw[->] (3.5,9.1) to[bend right] (8.5,9.1);
        
        \draw [-stealth](1.5,9.5) -- (2.5,9.5);
        
        \draw [-stealth](7.5,9.5) -- (8.5,9.5);

        \node at(-0.5,8.5) {{\color{white}$v_1$}};
        
        \node at (0.5,10.5) {$1$}; 
        \node at (1.5,10.5) {$2$}; 
        \node at (2.5,10.5) {{\color{OrangeRed}$3$}}; 
        
        \node at (3.5,9.5) {{\color{CornflowerBlue!50}{\huge$\symrook$}}}; 
        \node at (4.5,9.5) {{\color{CornflowerBlue!50}{\huge$\symrook$}}}; 
        \node at (5.5,9.5) {{\color{CornflowerBlue!50}{\huge$\symrook$}}}; 
        
        \node at (3.5,10.5) {$4$}; 
        \node at (4.5,10.5) {$5$}; 
        \node at (5.5,10.5) {$6$}; 
        
        \node at (6.5,9.5) {{\color{CornflowerBlue!50}{\huge$\symrook$}}}; 
        \node at (7.5,9.5) {{\color{CornflowerBlue!50}{\huge$\symrook$}}}; 
        \node at (8.5,9.5) {{\color{IndianRed}{\huge$\symrook$}}}; 
        
        \node at (6.5,10.5) {$7$}; 
        \node at (7.5,10.5) {$8$}; 
        \node at (8.5,10.5) {{\color{OrangeRed}$9$}};

        

        
        
        

        \end{tikzpicture}
         
        \caption{All blue rooks on the top row capture one of the red rooks leaving us in a solvable state with the two green rooks making the final captures.}
        \label{fig:fwd-rooks-step3}
     \end{subfigure}
        \caption{All illustration of the forward direction.}
        \label{fig:threegraphs}
\end{figure}

\paragraph*{Reverse Direction} Suppose the reduced instance is solvable. We first make some claims about any valid sequence of $s$ moves, denoted by $\sigma$, that clears the board. Let $\rho_\sigma((x,y),\ell)$ denote the type of the piece at the location $(x,y)$ after $\ell$ moves of $\sigma$ have been played. If $(x,y)$ is an empty location after $\ell$ moves of $\sigma$ have been played, then we let $\rho_\sigma((x,y),\ell) = \square$.

Let $\zeta(t)$ denote the set of locations $(2m+1,\cdot)$ occupied by red rooks on the top row of the board after $t$ moves of $\sigma$ have been made, in other words: $\zeta(t) = \{i ~|~ \rho_\sigma((2m+1,i),t) = {\color{IndianRed}\symrook}\}$.

\begin{claim}
$| \cup_{1 \leq t \leq s}\zeta(t)| \leq k$. 
\end{claim}

\begin{proof}
Let $i\in \underset{1\leq t\leq s}{\bigcup}\zeta(t)$. That is, there exists $1\leq t\leq s$ such that the cell $(2m+1,i)$ has a $0$-rook after $t$ moves of $\sigma$. Let $p>t$ denote the first move of $\sigma$ that empties the cell $(2m+1,i)$. Note that the cell $(2m+1,i)$ has a $1$-rook before the $p^{th}$ move of $\sigma$. So, there exists $t<q<p$ such that a $2$-rook captures $0$-rook at cell $(2m+1,i)$ in $q^{th}$ move of $\sigma$. Also, such a $2$-rook is one among the $k$ cleaner rooks. Thus, $\big|\underset{1\leq t\leq s}{\bigcup}\zeta(t)\big|\leq k$.
\end{proof} 

Let $j \in [m]$ and $i \in [n]$. We say that $i$ is an \emph{$j$-affected index} if there is some $t \in [s]$ such that the rook at position $(2j-1,i)$ was captured by the green rook originally at position $(2j-1,n+j)$ in the $t^{th}$ move of $\sigma$. 

\begin{claim}
For a fixed $j \in [m]$, there is exactly one $i \in [n]$ such that $i$ is a $j$-affected index.
\end{claim}

\begin{proof}
Let $j\in[m]$. Let $t\in[s]$ denote the first move of $\sigma$ wherein the collector 2-rook (say $g$) at the cell $(2j-1,n+j)$ either gets captured (Case 1) or captures (Case 2).\\\\
In Case 1, a terminal $1$-rook in the row $2j-1$ captures $g$ in the $t^{th}$ move of $\sigma$. After the $t^{th}$ move of $\sigma$, the cell $(2j-1,n+j)$ has a $0$-rook. Let $p>t$ denote the first move of $\sigma$ that empties the cell $(2j-1,n+j)$. Note that the cell $(2j-1,n+j)$ has a $1$-rook before the $p$-th move of $\sigma$. So, there exists $t<q<p$ such that a $2$-rook captures $0$-rook at the cell $(2j-1,n+j)$ in $q^{th}$ move of $\sigma$. However, no such $2$-rook exists on the board. Thus, Case 1 does not arise.\\\\
In Case 2, there exists $i\in[n]$ such that $g$ captures the rook at the cell $(2j-1,i)$ in the $t^{th}$ move of $\sigma$. Note that $i$ is the unique $j$-affected index.
\end{proof} 

We call $i \in [n]$ an \emph{affected index} if there is some $t \in [s]$ such that the position $(2m+1,i)$ was occupied by a red rook after $t$ moves of $\sigma$. 

\begin{claim}
If $i \in [n]$ is a $j$-affected index for some $j \in [m]$, then $i$ is also an affected index.
\end{claim}

\begin{proof}
Assume that $i\in[n]$ is a $j$-affected index for some $j\in[m]$. That is, there exists $t\in[s]$ such that in the $t^{th}$ move of $\sigma$, the collector $2$-rook at the cell $(2j-1,n+j)$ captures the rook at the cell $(2j-1,i)$. After the $t^{th}$ move of $\sigma$, the cell $(2j-1,i)$ has a $1$-rook. Let $t'>t$ denote the first move of $\sigma$ wherein the $1$-rook  at the cell $(2j-1,i)$ either gets captured (Case 1) or captures (Case 2).

In Case 1, a $1$-rook captures the $1$-rook at the cell $(2j-1,i)$ in the $t'^{th}$ move of $\sigma$. After the $t'^{th}$ move of $\sigma$, the cell $(2j-1,i)$ has a 0-rook. Let $p>t'$ denote the first move of $\sigma$ that empties the cell $(2j-1,i)$. Note that the cell $(2j-1,i)$ has a $1$-rook before the $p^{th}$ move of $\sigma$. So, there exists $t'<q<p$ such that a $2$-rook captures $0$-rook at the cell $(2j-1,i)$ in the $q^{th}$ move of $\sigma$. However, no such $2$-rook exists on the board. Thus, Case 1 does not arise.
\\\\
In Case 2, in the $t'^{th}$ move of $\sigma$, the $1$-rook at the cell $(2j-1,i)$ captures either rook at the cell $(2m+1,i)$ (Subcase 1), or rook at the cell $(2j'-1,i)$ for some $j'\in [m]\setminus\{j\}$ (Subcase 2). \\\\
In Subcase 1, the cell $(2m+1,i)$ has a $0$-rook after $t'$ moves of $\sigma$. So, $i$ is an affected index. \\\\
In Subcase 2, the cell $(2j'-1,i)$ has a $0$-rook after $t'$ moves of $\sigma$. Let $p'>t'$ denote the first move of $\sigma$ that empties the cell $(2j'-1,i)$. Note that the cell $(2j'-1,i)$ has a 1-rook before the $p'^{th}$ move of $\sigma$. So, there exists $t'<q'<p'$ such that a $2$-rook captures $0$-rook at the cell $(2j'-1,i)$ in the $q'^{th}$ move of $\sigma$. Note that this $2$-rook is the collector rook at the cell $(2j'-1,n+j')$. After the $q'^{th}$ move of $\sigma$, the cell $(2j'-1,i)$ has a $1$-rook. Let $t''>q$ denote the first move of $\sigma$ wherein the $1$-rook at the cell $(2j'-1,i)$ either gets captured or captures. As before, it can be argued that in the $t''^{th}$ move, the $1$-rook at the cell $(2j'-1,i)$ is not captured, and it either captures rook at the cell $(2m+1,i)$ (in which case we are done), or rook at the cell $(2j''-1,i)$ for some $j''\in[m]\setminus\{j,j'\}$ (in which case the collector $2$-rook at the cell $(2j''-1,n+j'')$ captures the $0$-rook at the cell $(2j''-1,i)$ in some subsequent move). Repeatedly using the same argument proves the claim. 
\end{proof}

Consider $S := \{\ell ~|~ \ell \in [n] \text{ and } \ell \text{ is an affected index}\}$. We claim that $S$ is a dominating set in $G$. Indeed, consider any non-terminal vertex $v_j \in B$. If $i$ is the unique $j$-affected index, then $i$ is also an affected index, and therefore belongs to the dominating set. Note that $i \in N(v_j)$ by construction, therefore we are done. Also, by Claim 6, we have that $|S| \leq k$. This concludes the proof in the reverse direction. 
\end{proof}

\subsection{Queens}
\label{sec:queens}

Recall the reduction described for the proof of~\Cref{thm:rooks}. It is straightforward to check that if we introduce a large number --- say $O(n^2)$ many --- empty columns between every pair of consecutive columns of the original board, then we can also replace the rooks by queens and the reduction will remain valid. This is because the vast empty spaces essentially ``nullify'' the additional diagonal moves of the queens, thereby reducing their behavior to being equivalent to rooks. Also note that the operation of adding empty columns does not affect the forward direction: all pairs of mutually attacking locations remain mutually attacking even after this modification.

We now present the following strengthening of this hardness result. Recall that in the previous reduction, we had pieces that were allowed to capture twice and others that were allowed to capture once. With queens, however, we can adapt the reduction so that every piece is allowed to capture twice, bringing this closer to the spirit of traditional solo chess:

\queens*

We note that this result can be achieved by replacing every queen that is allowed to capture once with the following pair of queens that are both allowed to capture twice, with the queen on the bottom right replacing the ``original'' $1$-queen:

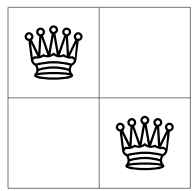
\begin{figure}[H]\centering
\centering
    $$~$$
    \begin{tikzpicture}
    
        
        \draw[step=1.2cm,color=gray] (0,0) grid (2.4,2.4);

        \node at (1.8,0.6) {{\Huge{$\symqueen$}}};
        \node at (0.6,1.8) {{\Huge{$\symqueen$}}};
        
\end{tikzpicture}
    \caption{The reduced instance of~\textsc{Generalized Solo Chess} played by rooks corresponding to the instance shown in~\Cref{fig:rbds}.}
    \label{fig:example1rednqueens}
\end{figure}


We call the queen on the top-left corner the supporting queen, and refer to the other queen as its partner. Once all the $1$-queens of the reduced instance are replaced in this way, we ensure that all supporting queens have the property that they do not attack any queen other than their partner. To achieve this, we shift them north-west along their diagonals appropriately if required. Note that the fact that the supporting queens attack only their partners forces that they are never captured by another piece, and that they capture their partner queen, which replaces the partner with a $1$-queen, as desired. 
We omit the details here.

\subsection{Pawns}
\label{sec:pawns}

In contrast to the cases of Rooks, Queens and Bishops, we show that \textsc{Generalized Solo Chess}(\sympawn,2) can be decided by an algorithm whose running time is linear in the number of pawns when all pawns are allowed to capture at most twice in the initial configuration.

\pawns*

We denote by $V$ the set of squares that initially contain a pawn, and by $t$ the location of the target pawn.
For $u,v \in V$, we say that $u$ is a {\it parent} of $v$ (and that $v$ is a {\it child} of $u$) if $u$ is diagonally one capture move away from $v$.  We denote by $C(v)$ the children of $v$.
Further if two vertices share a common parent, then we call them 
{\it siblings} of each other. We say that an initial configuration of pawns is \emph{super-solvable} if the final capturing pawn has one move remaining after the final capture.

\begin{definition}
We say that a configuration of \textsc{Generalized Solo Chess}($\sympawn,2)$ with position set $V$ is a skewed binary tree rooted at square $v$ if 
the following are true:
\begin{itemize}
\item[(a)] All pawns are on squares of the same color.
\item[(b)] All squares in $V \setminus \{v\}$ are below $v$.
\item[(c)] Every square in $V \setminus \{v\}$ has a parent in $V$.
\item[(d)] Every non-empty row below $v$ contains exactly two squares of $V$ with a common parent, except possibly the last (bottom-most) row which may contain one square of $V$.
\end{itemize}
\end{definition}
 
The following result is the key to the characterization of solvable instances.
\begin{lemma}\label{lem:pawns-1-solvable}
An instance of \textsc{Generalized Solo Chess}($\sympawn$,2) is super-solvable if and only if the initial configuration is a skewed binary tree.
\end{lemma}

\begin{figure}\centering
\scalebox{0.67}[0.67]{\includegraphics{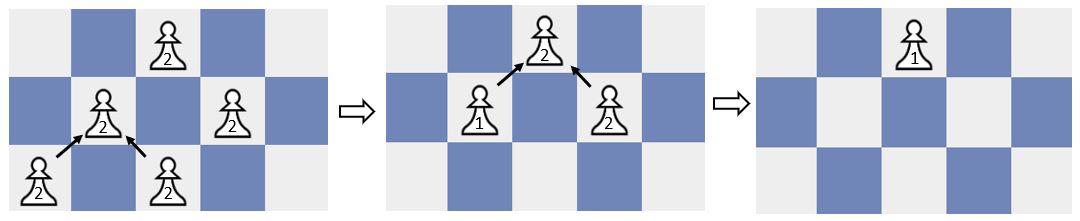}}
\caption{The initial configuration in this example is a skewed binary tree. Note that it is super-solvable because the shown sequence of moves clears the board such that the final pawn has one move left - here, the $2$-pawn at the cell $(2,4)$ does the final capture and becomes a $1$-pawn.}
\end{figure}

\begin{figure}\centering
\scalebox{0.67}[0.67]{\includegraphics{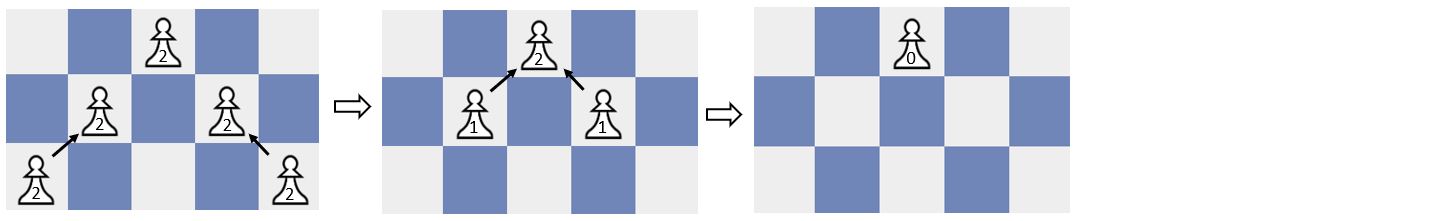}}
\caption{The initial configuration in this example is not a skewed binary tree. Note that it is not super-solvable because any sequence of moves that clears the board (one such sequence is shown) is such that the final pawn has no moves left.}
\end{figure}

In particular, we can verify in linear time whether a given configuration of \textsc{Generalized Solo Chess}($\sympawn$,2) is super-solvable, as each of the properties (a)-(d) can be checked in linear time. 

We first observe that pawns can capture only in the forward direction (upward for W pawns)  and only pawns on squares of the same color.
Thus, we shall henceforth assume that all pawns are on squares of the same color and also that there is exactly one pawn whose initial square has the largest $y$ co-ordinate; we shall call this the target pawn. If our assumption is false, we report the instance as a NO instance, and do not proceed further. We now describe the proof of~\Cref{lem:pawns-1-solvable}.

\begin{proof}
We prove the claim by induction on $|V|$.
If $|V|=1$, the instance is trivially super-solvable and also satisfies the definition of a skewed binary tree.
If $|V|=2$, then the instance is super-solvable if and only if the unique vertex $v \in V \setminus \{t\}$ is a child of $t$, and this configuration is a skewed binary tree.

Thus, we suppose that $|V| \geq 3$.
The necessity of conditions (a), (b), (c) has already been noted so that in the rest of this section we consider only configurations that satisfy (a), (b) and (c).
We shall now establish the necessity of condition (d).

Firstly, we claim that $t$ has two children.
Suppose not, and let $v$ be the only child of $t$. Then the last capture must be from $v$ to $t$, and the last but one capture must be at $v$, so that the token at $v$ has only one move remaining. When this token captures at $t$, it has zero moves left after the capture.

Thus, we can assume that $t$ has two children $u,v$.
Consider a valid super-solvable sequence $\sigma$; let $u$ be the vertex from which the final capture was made at $t$. Then the token at $u$ must have had two moves left before this capture and therefore no capture in $\sigma$ was ever made at $u$.
Also, the last but one capture in $\sigma$ must have been made from $v$ to $t$.
This implies that the sequence obtained from $\sigma$ by excluding the final capture is a valid super-solvable sequence for $V \setminus \{t,u\}$. Since $V \setminus \{t,u\}$ is super-solvable, by the induction hypothesis, property (d) holds; i.e. there are exactly two squares in every row below $v$, except possibly for the bottom-most non-empty row.
This shows that property (d) holds for $V$ as well.

For the other direction, suppose that a given configuration with $V$ as the set of squares is a skewed binary tree, and that $|V| \geq 3$.
Then by definition $t$ has two children $u,v$ and it must be the case that one of $u,v$, say $u$ has no child outside $C(v)$.
Then $V \setminus \{u,v\}$ must induce a skewed binary tree; let $\sigma$ be a super-solvable sequence for $V \setminus \{u,v\}$. Appending the captures $u \to t,v\to t$ yields a super-solvable sequence for the original configuration.

This completes the proof of Lemma \ref{lem:pawns-1-solvable}.
\end{proof}

We now proceed to the proof of Theorem \ref{thm-pawn-win}.
\begin{proof}
Let $V$ be the initial position set and $t$ be the target square.

{\bf Case 1:} $t$ has a single child $x$.
In this case, we note that the instance is solvable if and only if the configuration restricted to $V \setminus \{t\}$ with $x$ as target is super-solvable, which by Lemma \ref{lem:pawns-1-solvable} in linear time.

{\bf Case 2:} 
$t$ has two children $x,y$, and one of them, say $y$, has no child other than the common child of $x,y$.
In this case, the instance is solvable if and only if the configuration restricted to $V \setminus \{t,y\}$ with target $x$ is super-solvable, which we can verify in linear time.

{\bf Case 3:}
$t$ has two children $x,y$, and 
$|C(x) \cup C(y)|=3$; let $C(x) \cup C(y)=\{a,b,c\}$. Then the instance is solvable if and only if there's a re-labeling $u,v,w$ of $\{a,b,c\}$ such that $C(u) \cup C(v) \subseteq C(w)$ and the configuration restricted to $V \setminus \{t,x,y,u,v\}$ with target $w$ is super-solvable. This can again be verified in linear time.

This completes the proof of Theorem \ref{thm-pawn-win}.
\end{proof}



\section{Graph Capture Game}
\label{sec:graphs}
We introduce a game on graphs, which generalizes~\textsc{Generalized Solo Chess}($\symknight$,$d$) when played on undirected graphs and~\textsc{Generalized Solo Chess}($\sympawn$,$d$) when played on directed graphs.

\begin{tcolorbox}
\textsc{Graph Capture}(G,d):
\tcblower{}
{\bf Input:} A graph $G=(V,E)$.

{\bf Output:} Decide if there exists a sequence of token captures (along the edges of $G$) such that only a single token remains, with the constraint that each token may capture at most $d$ times.
\end{tcolorbox}

Our main result in this section is the following:

\graphcapture*

In the rest of this section, we say that $G$ is solvable if \textsc{Graph Capture}$(G,2)$ is a \textsc{Yes}-instance.

\subsection{Undirected Graphs}
\label{sec:gc-undirected}
We prove Theorem \ref{thm:graphcapture} for undirected graphs.

A rooted tree is a pair $(T,v)$, with $v$ denoting the root vertex; given a rooted tree $(T,v)$ and a vertex $w$ of $T$, we denote by $C(w)$ the children of $w$, 
and by $T(w)$ the subtree rooted at $w$.

\begin{lemma}\label{spanningTreeLemma}
A graph $G=(V,E)$ is solvable if and only $G$ contains a vertex $v$ and a spanning tree $T$ such that every internal node of the rooted tree $(T,v)$ has a leaf neighbor.
\end{lemma}

We prove Lemma \ref{spanningTreeLemma} in the Appendix and now turn to a proof of part (a) in~\Cref{thm:graphcapture}.

\begin{proof}
We proceed by a reduction from \textsc{Colorful Red-Blue Dominating Set}. Let $\langle G = (R \uplus B,E); k\rangle$ be an instance of \textsc{Colorful Red-Blue Dominating Set} with color classes $V_1 \cup \cdots \cup V_k$. We assume, without loss of generality, that $|V_1| = \cdots = |V_k| = n$ and let $V_j := \{v_1^{(j)}, \ldots, v_n^{(j)}\}$. We begin by describing the construction of the reduced instance. We begin with the graph $G$ and make the following additions:

\begin{enumerate}
    \item For all $i \in [n]$ and $j \in [k]$, introduce a vertex $u_i^{(j)}$ and make it adjacent to $v_i^{(j)}$. We call these the \emph{red partner vertices}.
    \item For each $u_i^{(j)}$, introduce two neighbors $p_i^{(j)}$ and $q_i^{(j)}$, and finally, introduce two vertices $r_i^{(j)}$ and $s_i^{(j)}$ that are adjacent only to $p_i^{(j)}$ and $q_i^{(j)}$ respectively. In other words, each $u_i^{(j)}$ has two degree two neighbors, which in turn have a leaf neighbor each. Combined, we refer to the collection of vertices $S_j := \{u_i^{(j)},p_i^{(j)},q_i^{(j)},r_i^{(j)},s_i^{(j)}~|~ i \in [n]\}$ as the \emph{selection gadget} for $V_j$. 
    
    \item For all $j \in [k]$, introduce a vertex $w_j$ and make it adjacent to $u_i^{(j)}$ for all $i \in [n]$. We call these vertices the \emph{guards}.
    \item We finally add a vertex $\star$ that is adjacent to all the red partner vertices. We also add the vertices $p$ and $q$ and the edges $(\star,p)$ and $(p,q)$.
\end{enumerate}

We let $H$ denote the graph thus constructed based on $G$ and ask if $H$ has a rooted spanning tree for which every internal note has a leaf neighbor. This completes a description of the construction. We briefly describe the intuition for the equivalence of the two instances. Because of the vertices selection gadgets, the partner vertices are forced to find their leaf neighbors in any spanning tree among the red vertices that they partner --- except for at most one, which can use the guard vertex as the leaf neighbor. This leads to one red vertex being left ``free'' of being a leaf neighbor to a partner vertex in each color class, hence we have a selection of one blue vertex per color class. Since these are the only possible entry points for the blue vertices into the spanning tree, the vertices ``chosen'' by the selection gadget must correspond to a dominating set. We now formalize this intuition.

\begin{figure}
\centering
    \includegraphics[scale=0.42]{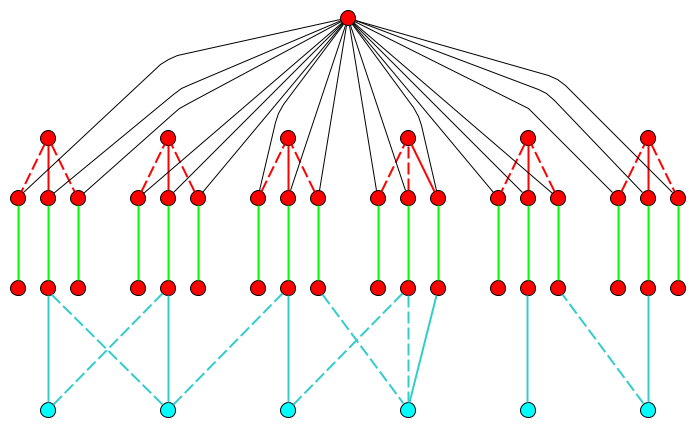}
    \caption{An illustration of the reduction from Red-Blue dominating set. The solid lines belong to the spanning tree. The ``spikes'' from the selection gadget are omitted for clarity.}
    \label{fig:graphGadget1}
\end{figure}

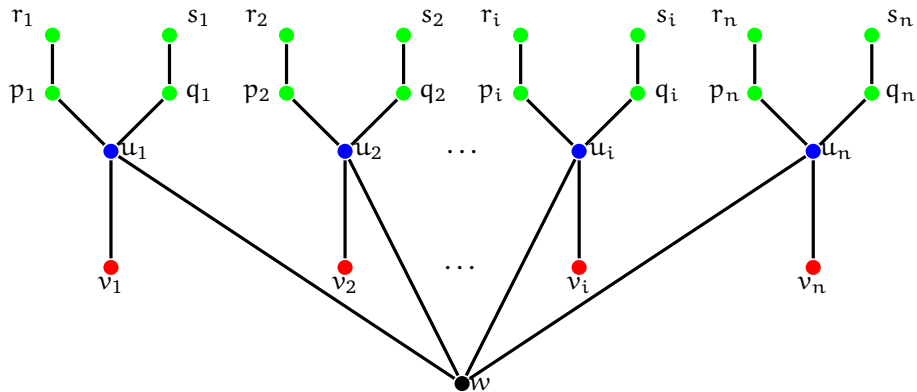
\begin{figure}
\centering
\begin{tikzpicture}[scale=.77,auto=center,every node/.style={circle,fill=black,inner sep=2pt},every path/.style={very thick}]
 

\node(w) at (8,0){};

\node(v1)[red] at (2,2){};
\node(v2)[red] at (6,2){};
\node(vi)[red] at (10,2){};
\node(vn)[red] at (14,2){};

\node(u1)[blue] at (2,4){};
\node(u2)[blue] at (6,4){};
\node(ui)[blue] at (10,4){};
\node(un)[blue] at (14,4){};

\node[fill=none] at (8.3,0) {$w$};

\node[fill=none] at (2,1.7){$v_1$};
\node[fill=none] at (6,1.7){$v_2$};
\node[fill=none] at (10,1.7){$v_i$};
\node[fill=none] at (14,1.7){$v_n$};

\node[fill=none] at (2.4,4){$u_1$};
\node[fill=none] at (6.4,4){$u_2$};
\node[fill=none] at (10.4,4){$u_i$};
\node[fill=none] at (14.4,4){$u_n$};

\draw[-](v1)--(u1);
\draw[-](v2)--(u2);
\draw[-](vi)--(ui);
\draw[-](vn)--(un);

\draw[-](w)--(u1);
\draw[-](w)--(u2);
\draw[-](w)--(ui);
\draw[-](w)--(un);

\node(p1)[green] at (1,5){};
\node(q1)[green] at (3,5){};
\draw[-](u1)--(p1);
\draw[-](u1)--(q1);
\node(r1)[green] at (1,6){};
\node(s1)[green] at (3,6){};
\draw[-](p1)--(r1);
\draw[-](q1)--(s1);

\node[fill=none] at (0.5,5){$p_1$};
\node[fill=none] at (3.5,5){$q_1$};
\node[fill=none] at (0.5,6.3){$r_1$};
\node[fill=none] at (3.5,6.3){$s_1$};

\node(p2)[green] at (5,5){};
\node(q2)[green] at (7,5){};
\draw[-](u2)--(p2);
\draw[-](u2)--(q2);
\node(r2)[green] at (5,6){};
\node(s2)[green] at (7,6){};
\draw[-](p2)--(r2);
\draw[-](q2)--(s2);

\node[fill=none] at (4.5,5){$p_2$};
\node[fill=none] at (7.5,5){$q_2$};
\node[fill=none] at (4.5,6.3){$r_2$};
\node[fill=none] at (7.5,6.3){$s_2$};

\node(pi)[green] at (9,5){};
\node(qi)[green] at (11,5){};
\draw[-](ui)--(pi);
\draw[-](ui)--(qi);
\node(ri)[green] at (9,6){};
\node(si)[green] at (11,6){};
\draw[-](pi)--(ri);
\draw[-](qi)--(si);

\node[fill=none] at (8.5,5){$p_i$};
\node[fill=none] at (11.5,5){$q_i$};
\node[fill=none] at (8.5,6.3){$r_i$};
\node[fill=none] at (11.5,6.3){$s_i$};

\node(pn)[green] at (13,5){};
\node(qn)[green] at (15,5){};
\draw[-](un)--(pn);
\draw[-](un)--(qn);
\node(rn)[green] at (13,6){};
\node(sn)[green] at (15,6){};
\draw[-](pn)--(rn);
\draw[-](qn)--(sn);

\node[fill=none] at (12.5,5){$p_n$};
\node[fill=none] at (15.5,5){$q_n$};
\node[fill=none] at (12.5,6.3){$r_n$};
\node[fill=none] at (15.5,6.3){$s_n$};

\node[fill=none] at (8,2){$\dots$};
\node[fill=none] at (8,4){$\ldots$};

\end{tikzpicture}
\caption{A schematic showing the selection gadget for one color class.}
    \label{fig:my_label}
\end{figure}

\paragraph*{Forward Direction}

Assume that $\langle G = (R \uplus B,E); k\rangle$ is a YES instance. That is, there exist $1\leq j_1,\ldots,j_k\leq n$ such that every vertex in $B$ has a neighbour in $\big\{v_{j_1}^{(1)},\ldots,v_{j_{k}}^{(k)}\big\}$. For each $b\in B$, let $f(b)$ denote an arbitrary but fixed $\ell\in[k]$ such that $v_{j_{\ell}}^{(\ell)}\in N(b)$.\\\\Let $T$ denote the spanning tree of $H$ with the following edge set:
\begin{itemize}
    \item $T$ contains all the edges of $H\Big[\{\star,p,q\}\cup R\cup \underset{1\leq\ell\leq k}{\bigcup} S_{\ell}\Big]$
    \item For every $b\in B$, $T$ contains the edge $\big\{v_{j_{f(b)}}^{(f(b))},b\big\}$
    \item For every $1\leq \ell\leq k$, $T$ contains the edge $\big\{u_{j_{\ell}}^{(\ell)},w_{\ell}\big\}$
\end{itemize}
Note that every internal node of the rooted tree $(T,\star)$ has a leaf neighbour.
\paragraph*{Reverse Direction}
Assume that there exist $r\in V(H)$ and a spanning tree (say $T$) of $H$ such that every internal node of the rooted tree $(T,r)$ has a leaf neighbour. In $T$, $\star$ is adjacent to $p$ and at least one red partner vertex. So, $\star$ has at least two neighbours in $T$. Thus, $\star$ is not a leaf node of $(T,r)$.\\\\The only neighbours of $\star$ in $T$ are $p$ and some vertices of $\{u_i^{(\ell)}~|~1\leq\ell\leq k,1\leq i\leq n\}$. Note that $p$ is not a leaf node of $(T,r)$ as $p$ has two neighbours, i.e., $\star$ and $q$, in $T$. Also, for every $1\leq \ell\leq k$ and every $1\leq i\leq n$, $u_i^{(\ell)}$ is not a leaf node of $(T,r)$ because $u_{i}^{(\ell)}$ has at least two neighbours, i.e., $p_i^{(\ell)}$ and $q_i^{(\ell)}$, in $T$. Therefore, $\star$ has no leaf neighbours in $(T,r)$. So, $\star$ is not an internal node of $(T,r)$. Hence, we have $r=\star$.\\\\
Let $1\leq \ell\leq k$. The only neighbours of $w_{\ell}$ in $T$ are some vertices of $\{u_{i}^{\ell}~|~1\leq i\leq n\}$. As argued earlier, no red partner vertex is a leaf node of $(T,r)$. So, $w_{\ell}$ has no leaf neighbours in $(T,r)$. Thus, $w_{\ell}$ is not an internal node of $(T,r)$. That is, $w_{\ell}$ is a leaf node of $(T,r)$. Hence, there exists a unique integer (say $g(\ell)$) in $[n]$ such that $u_{g(\ell)}^{(\ell)}$ is the neighbour of $w_{\ell}$ in $T$.\\\\
Now, it suffices to show that every vertex in $B$ has a neighbour in $\big\{v_{g(1)}^{(1)},\ldots,v_{g(k)}^{(k)}\big\}$. Let $b\in B$. There exist $1\leq\ell\leq k$ and $1\leq i\leq n$ such that $v_{i}^{(\ell)}\in N_T(b)$. As shown above, $u_{i}^{(\ell)}$ is not a leaf node of $(T,r)$. That is, $u_{i}^{(\ell)}$ is an internal node of $(T,r)$. So, $u_{i}^{(\ell)}$ has a leaf neighbour (say $z$) in $(T,r)$. No vertex of $V(T)\setminus\big\{p_i^{(\ell)}, q_i^{(\ell)}, \star, w_{\ell},v_{i}^{(\ell)}\big\}$ is adjacent to $u_{i}^{(\ell)}$ in $T$.
Note that 
\begin{itemize}
    \item $p_{i}^{(\ell)}$ is not a leaf node of $(T,r)$ as $p_i^{(\ell)}$ has at least two neighbours, i.e., $u_i^{(\ell)}$ and $r_{i}^{(\ell)}$, in $T$.
    \item $q_{i}^{(\ell)}$ is not a leaf node of $(T,r)$ as $q_i^{(\ell)}$ has at least two neighbours, i.e., $u_i^{(\ell)}$ and $s_{i}^{(\ell)}$, in $T$.
    \item If $v_i^{(\ell)}$ is a neighbour of $u_i^{(\ell)}$ in $T$, then $v_{i}^{(\ell)}$ is not a leaf node of $(T,r)$ because in such a case, $v_i^{(\ell)}$ has at least two neighbours, i.e., $b$ and $u_i^{(\ell)}$, in $T$.
\end{itemize}
Therefore, we have $z=w_{\ell}$ and hence, $i=g(\ell)$. 

This completes the argument for equivalence. 
\end{proof}

\subsection{Directed Acyclic Graphs}
\label{sec:gc-dags}
We now prove Theorem \ref{thm:graphcapture} for DAGs.

\begin{proof}
Let $\varphi$ be a given instance of 3-SAT, with clauses $C_1,C_2,\ldots,C_m$ over variables $x_1,x_2,\ldots,x_n$.

We construct the following directed graph $G=(W,E)$, where
$W=U \cup V \cup \{w\}$; $U=\{u_1,u_2,\ldots,u_m\}$
and $V=\{v_1,v_2,\ldots,v_n\} \cup \{v_1^T,v_2^T,\ldots,v_n^T\}
\cup \{v_1^F,v_2^F,\ldots,v_n^F\}$.
Intuitively, each vertex in $U$ represents a clause and vertices $v_i^T,v_i^F$ correspond to an assignment of $T,F$ respectively to $x_i$.

The edge set is $E=E_1 \cup E_2$,
where 
\begin{align*}
E_1=
&\{(u_i,v_j^T)|x_j \in C_i,1 \leq i \leq m,1 \leq j \leq n\} \\
&\cup \\
&\{(u_i,v_j^F)|\neg{x_j} \in C_i,1 \leq i \leq m, 1 \leq j \leq n\}
\end{align*}
and 
\begin{align*}
E_2=
\{(v_i^T,v_i)|1 \leq i \leq n\}
\quad
\cup
\quad
\{(v_i^F,v_i)|1 \leq i \leq n\}
\quad
\cup
\quad
\{(v_i,w)|1 \leq i \leq n\}.
\end{align*}

\begin{figure}
\centering
\begin{tikzpicture}[scale=.8,auto=center,every node/.style={circle,fill=black,inner sep=2pt},every path/.style={
        very thick}]
  \node (w) at (8,7){};
  \node (v1) at (3,6){};
  \node (v2) at (8,6){};
  \node (v3) at (13,6){};

  \draw[->] (v1)--(w);
  \draw[->] (v2)--(w);
  \draw[->] (v3)--(w);
 
 \node [fill=none] at (8,7.3) {$w$};
 \node [fill=none] at (3.6,5.8) {$v_1$};
 \node [fill=none] at (8.6,5.8) {$v_2$};
 \node [fill=none] at (13.6,5.8) {$v_3$};
 
  \node (v1T) at (2,5){};
  \node (v1F) at (4,5){};
  \node (v2T) at (7,5){};
  \node (v2F) at (9,5){};
  \node (v3T) at (12,5){};
  \node (v3F) at (14,5){};
  
 \node [fill=none] at (1.5,5.3) {$v_1^T$};
 \node [fill=none] at (4.5,5.3) {$v_1^F$};
 \node [fill=none] at (6.5,5.3) {$v_2^T$};
 \node [fill=none] at (9.5,5.3) {$v_2^F$};
 \node [fill=none] at (11.5,5.3) {$v_3^T$};
 \node [fill=none] at (14.5,5.3) {$v_3^F$};
  
  \draw[->] (v1T)--(v1);
  \draw[->] (v2T)--(v2);
  \draw[->] (v3T)--(v3);
  \draw[->] (v1F)--(v1);
  \draw[->] (v2F)--(v2);
  \draw[->] (v3F)--(v3);
 
  \node (u1) at (2,3){};
  \node (u2) at (6,3){};
  \node (u3) at (10,3){};
  \node (u4) at (14,3){};
 
 \node [fill=none] at (2,2.5) {$u_1$};
 \node [fill=none] at (6,2.5) {$u_2$};
 \node [fill=none] at (10,2.5) {$u_3$};
 \node [fill=none] at (14,2.5) {$u_4$};
 
  \draw[->,blue] (u1)--(v1T);
  \draw[->] (u1)--(v2F);
  \draw[->] (u1)--(v3T);
  \draw[->] (u2)--(v1T);
  \draw[->,blue] (u2)--(v2T);
  \draw[->] (u3)--(v1F);
  \draw[->,blue] (u3)--(v2T);
  \draw[->] (u3)--(v3F);
  \draw[->] (u4)--(v1F);
  \draw[->,blue] (u4)--(v2F);
\end{tikzpicture}
\caption{The DAG corresponding to the set of clauses 
$C_1=\{x_1,\neg{x_2},x_3\},
C_2=\{x_2,x_3\},
C_3=\{\neg{x_1},x_2,\neg{x_3}\},
C_4=\{\neg{x_2},\neg{x_3}\}$. The blue edges indicate the captures in Phase 1 for the satisfying assignment $x_1=T$, $x_2=T$, $x_3=F$.}

\end{figure}
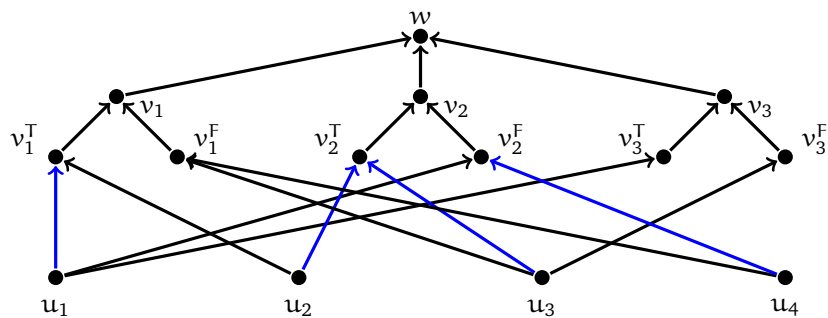

The graph $G$ can clearly be computed in time polynomial in the input size (number of variables and clauses).

It thus suffices to show that $\varphi$ is satisfiable iff $G$ is solvable.

First, suppose that $\varphi$ is satisfiable and let $A$ be a satisfying assignment for $\varphi$.
Consider the following sequence of captures:
\begin{itemize}
\item[Phase 1:] For each $i \in \{1,2,\ldots,m\}$, let $j$ be the least index such that $C_i$ is satisfied by $x_j$ or $\neg{x_j}$ in $A$.
Then the token at $u_i$ captures the token at $v_j^T$ (if $x_j=T$) or the token at $v_j^F$ (if $x_j=F$).
\item[Phase 2:] Now, for each $i \in \{1,2,\ldots,n\}$:
if $x_i=T$, then the token at $v_i^T$ captures the token at $v_i$, and then the token at $v_i^F$ captures the token at $v_i$;
otherwise $v_i=F$ and the token at $v_i^F$ captures the token at $x$, and then the token at $v_i^T$ captures the token at $v_i$;
\item[Phase 3:]
For each $i \in \{1,2,\ldots,n\}$,
the token at $v_i$ captures the token at $w$.
\end{itemize}
We show the validity of this capture sequence by considering each phase.

At the end of Phase 1, there are no tokens remaining at any of the $u_i$s, and further for each $i\in \{1,2,\ldots,n\}$, exactly one of the tokens among $\{v_i^T,v_i^F\}$ has 1 move remaining, and the other token has 2 moves remaining.

In Phase 2, the token among $\{v_i^T,v_i^F\}$ with 2 moves remaining is the last to capture at $v_i$; thus at the end of Phase 2, there is one token at each $v_i$ with one move remaining and further one more token at $w$; there are no other tokens.

In Phase 3, it is thus feasible for each token at $v_i$ to successively capture at $w$ and finally there is exactly one token remaining - at the vertex $w$.

Now, we suppose that $G$ is solvable; let $\sigma$ be a valid sequence of captures.
We claim that for each $i \in \{1,2,\ldots,n\}$, captures were not made at both $v_i^T$ and $v_i^F$ in $\sigma$.
For contradiction, suppose that captures were made both at $v_i^T$ and at $v_i^F$; let the last of these captures be at move $t_1$ of $\sigma$. 
Since all the tokens in $\{v_i^T,v_i^F\}$ must make their last capture at $v_i$, there must be a capture from $\{v_i^T,v_i^F\}$ to $z$ that appears in $\sigma$ later than $t$;
let the last such capture happen at move $t_2>t_1$. After move $t_2$, there are no tokens in $\{v_i^T,v_i^F\}$ and the token at $v_i$ has zero moves remaining; this implies that the token at $v_i$ cannot be cleared, which is the desired contradiction.

Now, let $I$ be the set of indices $i$ such that a capture was made at $v_i^T$.
Consider the assignment: for each $i \in \{1,2,\ldots,n\}$, we set $x_i=T$ if $i \in I$ and $x_i=F$ otherwise.
We claim that every clause is satisfied by this assignment.
Let $C_i$ be an arbitrary clause; if the token at $u_i$ made a capture at some $v_j^T$, then $C_i$ contains the literal $x_j$, and $j \in I$ so that $x_j=T$ and $C_i$ is satisfied.
If the token at $u_i$ made a capture at some $v_j^F$, then $C_i$ contains the literal $\neg{x_j}$. Also, by the claim in the previous paragraph, no capture was made at $v_j^T$, therefore $j \notin I$ and $x_j=F$, so that $C_i$ is satisfied.
\end{proof}

\section{Concluding Remarks}

We introduced~\textsc{Generalized Solo Chess} based on the Solo Chess game that is played on a $8 \times 8$ board. We focused mostly on scenarios that involve only pieces of one type, and showed that determining if a given instance is solvable is intractable when playing with rooks, bishops, and queens; while it is tractable for pawns. While we leave the case of knights open, we do show that a natural generalization of \textsc{Generalized Solo Chess} restricted to knights, \textsc{Graph Capture}, is hard even on DAGs and general undirected graphs. We also show that solvable instances of~\textsc{Generalized Solo Chess} played by rooks only admits an efficient characterization when the game is restricted to one-dimensional boards. 

Our work leaves open a few concrete open problems, which we enlist below:

\begin{enumerate}
    \item What is the complexity of \textsc{Generalized Solo Chess} played by rooks only, for the special case when all rooks are allowed to capture at most twice initially? Notice that if we replace rooks by queens in this question, we show \textsf{NP}-completeness (\Cref{thm:queens}).
    \item What is the complexity of \textsc{Generalized Solo Chess} played by pawns only, when the pawns are allowed at most a designated number of captures? Recall that if all pawns can capture at most twice, we have an efficient characterization (\Cref{thm-pawn-win}). 
    \item  What is the complexity of \textsc{Generalized Solo Chess} played by knights only?
\end{enumerate}
There are also several broad directions for future work, and we suggest some that we think are both natural and interesting problems to consider:

\begin{enumerate}
    \item For \textsc{NO}-instances of \textsc{Generalized Solo Chess}, a natural optimization objective is to play as many moves as possible, or, equivalently, leave as few pieces as possible on the board. It would also be interesting to find the smallest $d$ for which a board can be cleared if every piece was allowed to capture at most $d$ times. These are natural optimization versions that we did not explicitly consider but we believe would be interesting to explore. 
    \item Does~\textsc{Generalized Solo Chess} become easier if the number of pieces in every row or column is bounded? Note that the reduction in~\Cref{thm:rooks} can be used to show that \textsc{Generalized Solo Chess}($\symrook,2)$ is \textsf{NP}-complete even when the number of rooks per column is a constant, if we initiate the reduction from an instance of \textsf{Red-Blue Dominating Set} where every red vertex has constant degree.
    \item What is the complexity of~\textsc{Generalized Solo Chess} when played on boards of dimension $M \times N$, where one of $M$ or $N$ is a constant? It is not hard to generalize~\Cref{thm:rooks1d} to $2 \times N$ boards, however, a general result  --- say parameterized by one of the dimensions --- remains open.
    \item Variants of~\textsc{Generalized Solo Chess} where the pieces are limited not by the number of captures but the total distance moved on the board, or the distance moved per step, are also interesting to consider. Note that if the distance moved per step is lower bounded, then this forbids ``nearby captures'', while if it is upper bounded, then ``faraway captures'' are disallowed. 
    \item It would also be interesting to restrict the number of capturing pieces instead of the number of captures per piece. For example, it seems intuitive to posit that if we are permitted only one capturing piece, then it must trace a ``Hamiltonian path'' of sorts among the pieces on the board.
    \item Finally, we did not explicitly study \textsc{Generalized Solo Chess} with multiple piece types on the board --- although most such variants would be hard given the complexity results established already for pieces of one kind, it would be interesting to investigate special cases and exact algorithms in this setting.
\end{enumerate}



\bibliography{refs}

\newpage 

\appendix

\section{Rooks on 1D boards with $d=2$}


\begin{lemma}\label{lem:Rooks1D}
For every configuration $s$ of length $N$:
\begin{enumerate}
    \item $s$ is $N$-solvable iff $s[N]\neq\square$ and $s[1,\ldots,N-1]$ has at least as many $2$'s as $0$'s.
    \item $s$ is $1$-solvable iff $s[1]\neq\square$ and $s[2,\ldots,N]$ has at least as many $2$'s as $0$'s.
\end{enumerate}
\end{lemma}
\begin{proof}
\begin{enumerate}
    \item We first argue the forward implication. Let $s$ be a configuration of length $N$ such that $s$ is $N$-solvable. That is, there exists a sequence of moves (say $\sigma$) that clears the corresponding $1\times N$ board with the final rook at the cell $(1,N)$. Clearly, $s[N]\neq\square$. Let $1\leq i<N$ such that $s[i]=0$. That is, the cell $(1,i)$ has a $0$-rook before the first move of $\sigma$. Let $t$ denote the first move of $\sigma$ that empties the cell $(1,i)$. Note that the cell $(1,i)$ has a $1$-rook before the $t^{th}$ move of $\sigma$. So, there exists $1\leq q<t$ such that a $2$-rook captures $0$-rook at the cell $(1,i)$ in the $q^{th}$ move of $\sigma$. Also, before the $q^{th}$ move of $\sigma$, such a $2$-rook is located at $(1,j)$ for some $1\leq j<N$ such that $s[j]=2$. Thus, $s[1,\ldots,N-1]$ has at least as many $2$'s as $0$'s.

For the converse, we show (using induction on $m$) that the following statement is true for all integers $m\geq 0$: For every configuration $s$ such that $s[N]\neq\square$, if $s[1,\ldots,N-1]$ has exactly $m$ $0's$ and has at least as many $2's$ as $0$'s, then $s$ is $N$-solvable. 

For the base case, consider $m=0$. Let $s$ be a configuration such that $s[N]\neq\square$ and $s[1,\ldots,N-1]$ has no $0$'s. For each $1\leq i< N$, the cell $(1,i)$ in the corresponding board is either empty or has a $1/2$-rook. The board can be cleared with the final piece at $(1,N)$ by making the $1/2$-rooks (if any) to capture rook at the cell $(1,N)$. Thus, $s$ is $N$-solvable. So, the statement is true for $m=0$.\\\\
As induction hypothesis, assume that the statement is true for all integers $0\leq m\leq p$, for some integer $p\geq 0$. Let's argue that the statement is true for $m=p+1$. Let $s$ be a configuration such that $s[N]\neq\square$, and $s[1,\ldots,N-1]$ has exactly $(p+1)$ $0's$ and at least as many $2$'s as $0$'s. While there is a $1$-rook in the cells $(1,1),\ldots,(1,N-1)$ that can capture a $0$-rook in the cells $(1,1),\ldots,(1,N-1)$, make such a move. Once no such move can be made, there exist integers $1\leq u<v< N$ such that $s[u,\ldots,v]= 0\square^{\lambda}2$ or $s[u,\ldots,v]=2\square^{\lambda}0$, for some integer $\lambda\geq 0$. In the former (resp. latter) case, the $2$-rook at the cell $(1,v)$ (resp. $(1,u)$) can be made to capture the $0$-rook at the cell $(1,u)$ (resp. $(1,v)$), and the configuration corresponding to the resulting board is $N$-solvable by induction hypothesis.
\item Let $s^{rev}$ denote the string obtained by reversing $s$. Note that $s$ is $1$-solvable iff $s^{rev}$ is $N$-solvable. Also, using $1.$, $s^{rev}$ is $N$-solvable iff $s^{rev}[N]\neq\square$ and $s^{rev}[1,\ldots N-1]$ has at least as many $2$'s as $0$'s.
\end{enumerate}
This concludes the proof. 
\end{proof}
Thus, when $d = 2$, a configuration $s$ of length $N$ is solvable iff there exists $1\leq\ell\leq N$ such that
\begin{itemize}
    \item $s[\ell]\neq\square$,
    \item $s[1,\ldots,\ell-1]$ has at least as many $2$'s as $0$'s, and
    \item $s[\ell+1,\ldots,N]$ has at least as many $2$'s as $0$'s.
\end{itemize}

\newpage 

\section{Proof of Lemma \ref{spanningTreeLemma}}

We observe that for a graph $G$ which is a Yes Instance, the set of edges along which captures are made must induce a connected subgraph of $G$, and further the number of such edges is equal to $|V|-1$, and therefore these edges induce a spanning tree.

For $i \in \{0,1\}$, we say that a rooted tree $(T,v)$ is $i$-solvable if there is a sequence of captures such that (a) the final is made at $v$; and (b) the final capture piece has at least $i$ moves left after the final capture at $v$. We remark that this notation as used for graphs is different from that used for configurations of rooks.

Let $T_d$ denote the set of trees which are Yes Instances of GraphCapture(G,d).
Then an arbitrary graph $G$ is a Yes Instance of GraphCapture(G,d) if and only if $G$ contains a spanning tree in $T_d$. The class $T_1$ consists precisely of the star graphs.

\begin{proof}
We shall prove Lemma \ref{spanningTreeLemma} using the three claims that follow this proof.

From our earlier observation, it suffices to show the following:
If $T$ is a tree, then $T \in T_2$ if and only if $T$ has a vertex $v$ such that every internal node of $(T,v)$ has a leaf neighbor.

First, let $T \in T_2$. 
Let $v$ be the vertex where the final capture is made in a valid capture sequence.
Then $(T,v)$ is 0-solvable, and by Claim \ref{solvableClaim} every internal node of $(T,v)$ has a leaf neighbor.

For the other direction, note that by Claim \ref{solvableClaim}, we have that $(T,v)$ is 0-solvable and in particular, $T \in T_2$.
\end{proof}

\begin{claim}\label{zeroSolvableClaim}
Let $(T,v)$ be a rooted tree.
Then $(T,v)$ is 0-solvable
$\iff$
$\forall w \in C(v)$, the subtree $(T(w),w)$ is 1-solvable.
\end{claim}

\begin{claim}\label{oneSolvableClaim}
Let $(T,v)$ be a rooted tree with at least two vertices.
Then $(T,v)$ is 1-solvable
$\iff$ $v$ has a leaf neighbor and 
$\forall w \in C(v)$, the subtree $(T(w),w)$ is 1-solvable.
\end{claim}

\begin{claim}\label{solvableClaim}
Let $(T,v)$ be a rooted tree.
Then we have:

(a) $(T,v)$ is 0-solvable
$\iff$ every internal vertex of $(T,v)$ has a leaf neighbor;

(b) $(T,v)$ is 1-solvable
$\iff$ $|T|=1$ OR $(T,v)$ is 0-solvable and $v$ has a leaf neighbor.
\end{claim}

\begin{proof}
We prove (a) and (b) together by induction on $|V(T)|$.

When $|V(T)|=1$ and $|V(T)=2$, the statements are easily verified to be true.

Now we proceed to the induction step.

\begin{enumerate}
\item
(a) $\Rightarrow$:
Suppose that $(T,v)$ is 0-solvable.
By Claim \ref{zeroSolvableClaim}, for each $w \in C(v)$, we have that $(T(w),w)$ is 1-solvable.
In particular, each $(T(w),w)$ is 0-solvable, and thus by the induction hypothesis (a), every internal node not in $C(w)$ has a leaf neighbor.
By Claim \ref{oneSolvableClaim}, we also know that every $w \in C(v)$ must also have a leaf neighbor (if $|T(w)| \geq 2)$ or itself be a leaf vertex (if $|T(w)|=1)$.
This shows that every internal vertex of $(T,v)$ has a leaf neighbor.

\item
(a) $\Leftarrow$:
Suppose that $T$ has a vertex $v$ such that every internal vertex of $(T,v)$ has a leaf neighbor.
By the induction hypothesis (a), we know that for each $w \in C(w)$ the subtree $T(w)$ is 0-solvable. Further, we have either $|T(w)|=1$ or $w$ (being an internal vertex) must have a leaf neighbor. Thus, by induction hypothesis (b), we obtain that $(T(w),w)$ is 1-solvable for every $w \in C(v)$ and hence by Claim \ref{zeroSolvableClaim}, we obtain that $(T,v)$ is 0-solvable.

\item
(b) $\Rightarrow$:
Suppose that $(T,v)$ is 1-solvable.
Then $(T,v)$ is 0-solvable and further, either $|T|=1$, or by Claim \ref{oneSolvableClaim}, we get that $v$ has a leaf neighbor.

\item
(b) $\Leftarrow$:
Suppose that $(T,v)$ is 0-solvable and that $v$ has a leaf neighbor $w$.
Consider a valid sequence $\sigma$ that results in the last capture being made at $v$. One of these captures must be from $w$ to $v$. Moving this capture to the end of $\sigma$ maintains the validity of the sequence since $w$ has no children, and results in the final token at $w$ having at least one move remaining.
\end{enumerate}
This concludes the proof of the claim.\end{proof}

\end{document}